\documentclass[12pt]{article}

\usepackage{a4wide}

\usepackage{xspace} 
\usepackage{eurosym,pdflscape} 

\usepackage{calc} 
\usepackage{url}
\usepackage{color}
\usepackage{rotating}

\usepackage{tikz}
\usetikzlibrary{arrows,patterns,topaths}
\tikzstyle{mundo}=[circle,draw,inner sep=0pt,minimum size=25pt]
\tikzstyle{mundomedioblanco}=[white,circle,draw,draw opacity=0,inner sep=0pt,minimum size=18pt]
\tikzstyle{mundomedio}=[circle,draw,inner sep=0pt,minimum size=18pt]
\tikzstyle{mundomediogris}=[fill=lightgray,circle,draw,inner sep=0pt,minimum size=18pt]
\tikzstyle{mundoaction}=[circle,draw,inner sep=0pt,minimum size=12pt]
\tikzstyle{conjunto}=[inner sep=0pt,minimum size=16pt]
\tikzstyle{texto}=[inner sep=0pt,minimum size=10pt]

\usepackage{graphicx}

\usepackage{array} 

\usepackage{amsmath}
\usepackage{amssymb} 
\DeclareMathAlphabet{\mathpzc}{OT1}{pzc}{m}{it} 
\usepackage{mathrsfs} 

\usepackage{fancyvrb}

\usepackage[normalem]{ulem}
\usepackage{amsmath}

\definecolor{TODO_COLOR}{rgb}{1,0.5,0.5}

\newcommand{\union}{\cup}
\newcommand{\Union}{\bigcup}

\newcommand{\inter}{\cap}

\newcommand{\limply}{\rightarrow}

\newcommand{\imp}{\limply}

\renewcommand{\phi}{\varphi}

\newtheorem{theorem}{Theorem}

\newtheorem{proposition}[theorem]{Proposition}

\newtheorem{definition}[theorem]{Definition}

\newtheorem{lemma}[theorem]{Lemma}
\newtheorem{corollary}[theorem]{Corollary}

\usepackage{newproof}

\newcommand{\agents}{A}

\newcommand{\identity}{I}

\newcommand{\Nat}{\mathbb{N}}

\newcommand{\protocol}{\mathsf{P}}
\newcommand{\CO}{\mathsf{CO}}
\newcommand{\wCO}{\mathsf{wCO}}
\newcommand{\SFS}{\mathsf{ANY}}
\newcommand{\LNS}{\mathsf{LNS}}

\newcommand{\TOK}{\mathsf{TOK}}
\newcommand{\SPI}{\mathsf{SPI}}

\newcommand{\degr}{\mathit{deg}}
\newcommand{\outdeg}{\degr_{\mathsf{out}}}
\newcommand{\indeg}{\degr_{\mathsf{in}}}
\newcommand{\iodegr}{\degr_{\mathsf{io}}}

\newcommand{\extension}[2]{{#1_#2}}

\newcommand{\weg}[1]{}

\begin{document}

\title{Dynamic Gossip}
\author{Hans van Ditmarsch\thanks{LORIA, CNRS -- University of Lorraine, France \& ReLaX UMI 2000, Chennai, India} \and Jan van Eijck\thanks{University of Amsterdam \& CWI, Netherlands} \and Pere Pardo\thanks{Ruhr-Universit\"{a}t Bochum, Germany} \and Rahim Ramezanian\thanks{Sharif University of Technology, Iran} \and Fran\c{c}ois Schwarzentruber\thanks{ENS Rennes, IRISA, France}}
\date{\today}

\maketitle

\begin{abstract}
A gossip protocol is a procedure for spreading secrets among a group of agents, using a connection graph. The goal is for all agents to get to know all secrets, in which case we call the execution of the protocol successful. We consider distributed and dynamic gossip protocols. In distributed gossip the agents themselves instead of a global scheduler determine whom to call. In dynamic gossip not only secrets are exchanged but also telephone numbers (agent identities). This results in increased graph connectivity. We define six such distributed dynamic gossip protocols, and we characterize them in terms of the topology of the graphs on which they are successful, wherein we distinguish strong success (the protocol always terminates, possibly assuming fair scheduling) from weak success (the protocol sometimes terminates). For five of these protocols strong (fair) and weak success are characterized by weakly connected graphs. This result is surprising because the protocols are fairly different. In the sixth protocol an agent may only call another agent if it does not know the other agent's secret. Strong success for this protocol is characterized by graphs for which the set of non-terminal nodes is strongly connected. Weak success for this protocol is characterized by weakly connected graphs satisfying further topological constraints that we define in the paper. One direction of this characterization is surprisingly harder to prove than the other results in this contribution.
\end{abstract}


\section{Introduction}

\paragraph*{Gossip}
Gossip protocols are procedures for spreading secrets among a group of agents in a network wherein the agents are represented by the nodes, and the connections represent the ability of the agents to contact each other. There is a vast literature on gossiping in networks \cite{hedetniemietal:1988,hromkovicetal:2005,eugster,doerr,rajaraman,lewin}, and there are relations to the study of the spreading of epidemics \cite{kermarrecetal:2007}. In the original set-up \cite{tijdeman:1971, hajnal:1972} the network is a complete graph, which means that all agents can call each other, i.e., all agents have the telephone numbers of all other agents, and one of the key questions was to find a minimal sequence of calls to achieve a state where all agents know all secrets. On the assumption that during a call between two agents they exchange all their secrets, in a complete graph with $n > 3$ agents, $2n - 4$ calls are needed for this. Assuming four agents $a,b,c,d$, a call sequence for ensuring that all secrets are shared is $ab;cd;ac;bd$ (where $ab$ represents a call from $a$ to $b$, etc.). The secrets can be shared in $4$ calls. If a fifth agent $e$ is also present, precede this sequence with $ae$, and close off with $ae$: $6$ calls. In general, two extra calls are sufficient for each additional agent, and we have that the number of calls for $n$ agents equals $2(n-4) + 4 = 2n - 4$. Less than $2n-4$ calls is not possible \cite{tijdeman:1971,hurkens:2000}. On other than complete graphs the minimum number of calls needed to distribute all secrets may be larger \cite{harary,golumbic}, where we can distinguish undirected graphs (if I know your number, you know my number) from directed graphs (I may know your number, but you may not know my number). We consider undirected graphs. Instead of {\em exchange} of secrets, also known as {\em push-pull}, another setting is where the caller only {\em informs} the agent called ({\em push}), or only {\em receives} information from the person called ({\em pull}). We only consider push-pull. 

\paragraph*{Distributed gossip}
In the case above, the gossip procedure is regulated by a central authority, but in distributed computing we look for procedures that do not need such outside regulation. Distributed gossip protocols are standard in the networks community \cite{eugster,hromkovicetal:2005,kermarrecetal:2007}, and in an epistemic logical setting distributed gossip protocols were proposed in \cite{maduka-ecai,apt-tark}. Two such protocols are:
\begin{itemize}
\item $\SFS$ \quad
Until every agent knows all secrets, choose different agents $x$ and $y$ such that $x$ knows the number of $y$, and let $x$ call $y$.
\item $\LNS$ \quad 
Until every agent knows all secrets, choose different agents $x$ and $y$ such that $x$ knows the number of $y$ but not the secret of $y$, and let $x$ call $y$. 
\end{itemize}
Agent $x$ knows the number of agent $y$ iff node $y$ is a neighbour/successor of node $x$. In Protocol $\SFS$, {\em A}ny call can be made between neighbouring nodes, whereas in protocol $\LNS$, agent $x$ has to {\em L}earn a {\em N}ew {\em S}ecret in the call. Protocols like $\SFS$  are extremely standard in network theory \cite{kermarrecetal:2007}, protocols like $\LNS$ are found in \cite{maduka-ecai,apt-tark,Haeupler15}.

In this formulation as ``{\em Until (termination condition) is satisfied, choose different agents $x$ and $y$ such that (execution condition) is satisfied, and let $x$ call $y$}'' the distributed nature of the gossip protocol is absent. The distributed nature appears if we describe it as follows. \begin{quote} {\em Each agent $x$ runs the following program: choose another agent $y$ such that (execution condition) is satisfied, and call $y$. The environment runs the following program: until (termination condition) is satisfied, choose an agent $x$ (i.e., choose a program for agent $x$).} \end{quote}
Even this informal description leaves two aspects of gossip protocols that are not distributed, namely the selection of an agent $x$ allowed to make a call, and the part of the termination condition that involves checking whether all agents know all secrets. The selection of an agent (program) $x$ can be thought of as randomly unblocking the telephone of a single agent for an outgoing call. So this is a move of nature. The termination condition is not a distributed feature because we may not assume that an agent knows how many agents there are, therefore it cannot know that it knows all secrets (it cannot determine that the set of secrets it has stored is the set of all secrets). Instead of a protocol with a termination condition we may just as well see this as a protocol reaching a stable network configuration, i.e., a stable distribution of numbers and secrets over agents. Termination is then simply the moment that  stabilization is reached.

In gossip works calls are often made in rounds of parallel calls. This is already found in 1970s classics such as \cite{knoedel:1975} and treated in detail in the survey \cite{hedetniemietal:1988}, in the interpretation wherein disjoint pairs of agents are simultaneously selected to make calls: calls $xy$ and $zw$ can only take place in parallel if $\{x,y\}\inter\{z,w\} = \emptyset$: an agent cannot receive multiple calls. In more recent work involving rounds of calls \cite{eugster,doerr,rajaraman,lewin,kermarrecetal:2007}, a round means that all agents simultaneously make calls, but that agents may receive multiple calls from other agents. In that setting, the agent called gives the agent calling the information available prior to receiving the incoming calls in that round \cite{Haeupler15} (and not `on the go', while receiving more and more incoming calls). This requires all agents to know when a round starts and ends, which is not a distributed feature. Note that it leads to secrets distributions that are impossible in our setting. For example, given three agents $a,b,c$, then after the round consisting of the three calls $ac,bc,cb$, agent $a$ knows the secret of $c$, agent $b$ knows the secret of of $c$, and agent $c$ it the single agent to know all three secrets. Whereas after sequences of calls, the largest number of secrets is always known by at least 2 agents.

Questions about protocol execution length get different answers in distributed gossip. For 4 agents, the already mentioned $ab;cd;ac;bd$ is a length 4 execution sequence of $\LNS$, but $\LNS$ also has an execution $ab;ac;ad;bc;bd;cd$ of length 6. For $n$ agents, any sequence between the minimum of $2n-4$ and the maximum of (all possible pairs) $n(n-1)/2$ can be realized \cite{maduka-ecai}. In a distributed protocol we cannot guarantee any of these in advance. But all the executions of $\LNS$ are finite, whereas the $\SFS$ protocol has infinite executions, such as $ab;ab;\dots$. In the long run this does not matter: the expected execution of $\SFS$ is in the order of $n \log n$ \cite{boydSteele1979,haigh81,moon72}, and this is also the case for $\LNS$ and for yet other protocols considered in this work \cite{DitmarschKS17}.

\paragraph*{Dynamic gossip}
In {\em dynamic gossip} we assume that when a call is established not only secrets are exchanged but also numbers. The structure of the gossip graph constrains what calls can be made, so if also numbers are exchanged in calls, the connectivity of the graph may increase. We restrict our scope by assuming that in each call {\em all} secrets and {\em all} numbers held by the callers are exchanged. Such network changing protocols are a standard feature in work on gossip, epidemics, and resource discovery \cite{rajaraman,eugster,lewin,Haeupler15}. It is then common to consider protocols that investigate either exchange of secrets (or of other message content), or exchange of telephone numbers (or of other identifying information, such as node ID, IP adresses). We do not know of approaches that combine dynamic gossip and distributed gossip for protocols other than making random calls (i.e., protocol $\SFS$). Other works also consider {\em reducing} connectivity \cite{kermarrecetal:2007}, unlike us.

Questions about execution length also get different answers in dynamic gossip, and this is already the case in non-distributed dynamic gossip. An old result is that a circle network of $n \geq 5$ nodes requires at least $2n-3$ calls, as any network not containing a $4$-cycle requires more than $2n-4$ calls \cite{bumby:1981}. For example, given nodes $a,b,c,d,e$, mutually connected in that order and with $e$ connected to $a$, the minimum length of 7 is obtained if we go round until all agents know all secrets: $ab; bc; cd; de; ea; ab; bc$. But in dynamic gossip, 6 calls $ab; cd; ea; de; ac; bc$ are sufficient; it can be easily shown that $2n-4$ is the minimum for any circle. We will only incidentally address execution length in this work, to illustrate the difference between protocols. 

\paragraph*{Strong and weak success}
For an example of distributed dynamic gossip, let us consider three agents only, and that agents initially only know their own secret, and initially at least know their own number. Let us consider again the protocol $\LNS$, however with the change that now also numbers are exchanged in a call.

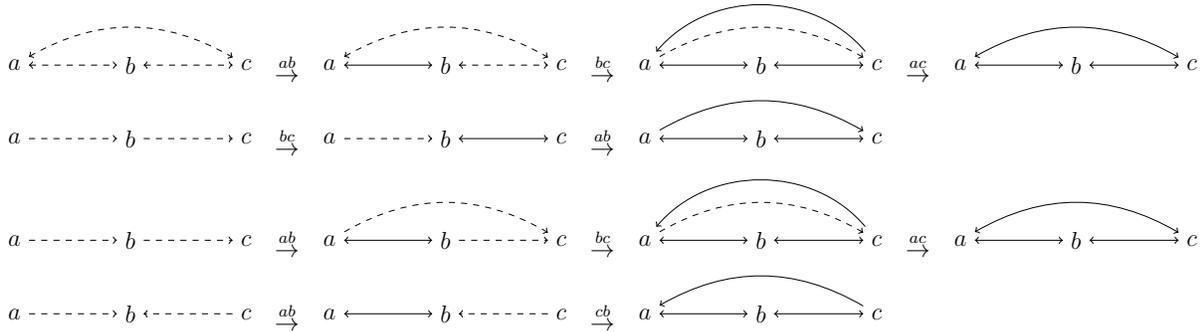
\begin{figure}
\hspace{-.4cm} \scalebox{.77}{
\begin{tabular}{l}
{%
\begin{tikzpicture}
\node (a) at (0,0) {$a$};
\node (b) at (2,0) {$b$};
\node (c) at (4,0) {$c$};
\path[dashed,<->] (a) edge (b);
\path[dashed,<->] (b) edge (c);
\path[dashed,<->,bend left] (a) edge (c);
\end{tikzpicture}
~$\stackrel {ab} \imp$~
\begin{tikzpicture}
\node (a) at (0,0) {$a$};
\node (b) at (2,0) {$b$};
\node (c) at (4,0) {$c$};
\path[<->] (a) edge (b);
\path[dashed,<->] (b) edge (c);
\path[dashed,<->,bend left] (a) edge (c);
\end{tikzpicture}
~$\stackrel {bc} \imp$~
\begin{tikzpicture}
\node (a) at (0,0) {$a$};
\node (b) at (2,0) {$b$};
\node (c) at (4,0) {$c$};
\path[<->] (a) edge (b);
\path[<->] (b) edge (c);
\path[dashed,->,bend left] (a) edge (c);
\path[->,bend right=50] (c) edge (a);
\end{tikzpicture}
~$\stackrel {ac} \imp$~
\begin{tikzpicture}
\node (a) at (0,0) {$a$};
\node (b) at (2,0) {$b$};
\node (c) at (4,0) {$c$};
\path[<->] (a) edge (b);
\path[<->] (b) edge (c);
\path[<->,bend left] (a) edge (c);
\end{tikzpicture}
}%
\\
{%
\begin{tikzpicture}
\node (a) at (0,0) {$a$};
\node (b) at (2,0) {$b$};
\node (c) at (4,0) {$c$};
\path[dashed,->] (a) edge (b);
\path[dashed,->] (b) edge (c);
\end{tikzpicture}
~$\stackrel {bc} \imp$~
\begin{tikzpicture}
\node (a) at (0,0) {$a$};
\node (b) at (2,0) {$b$};
\node (c) at (4,0) {$c$};
\path[dashed,->] (a) edge (b);
\path[<->] (b) edge (c);
\end{tikzpicture}
~$\stackrel {ab} \imp$~
\begin{tikzpicture}
\node (a) at (0,0) {$a$};
\node (b) at (2,0) {$b$};
\node (c) at (4,0) {$c$};
\path[<->] (a) edge (b);
\path[<->] (b) edge (c);
\path[->,bend left] (a) edge (c);
\end{tikzpicture}
}%
\\
{%
\begin{tikzpicture}
\node (a) at (0,0) {$a$};
\node (b) at (2,0) {$b$};
\node (c) at (4,0) {$c$};
\path[dashed,->] (a) edge (b);
\path[dashed,->] (b) edge (c);
\end{tikzpicture}
~$\stackrel {ab} \imp$~
\begin{tikzpicture}
\node (a) at (0,0) {$a$};
\node (b) at (2,0) {$b$};
\node (c) at (4,0) {$c$};
\path[<->] (a) edge (b);
\path[dashed,->] (b) edge (c);
\path[dashed,->,bend left] (a) edge (c);
\end{tikzpicture}
~$\stackrel {bc} \imp$~
\begin{tikzpicture}
\node (a) at (0,0) {$a$};
\node (b) at (2,0) {$b$};
\node (c) at (4,0) {$c$};
\path[<->] (a) edge (b);
\path[<->] (b) edge (c);
\path[dashed,->,bend left] (a) edge (c);
\path[->,bend right=50] (c) edge (a);
\end{tikzpicture}
~$\stackrel {ac} \imp$~
\begin{tikzpicture}
\node (a) at (0,0) {$a$};
\node (b) at (2,0) {$b$};
\node (c) at (4,0) {$c$};
\path[<->] (a) edge (b);
\path[<->] (b) edge (c);
\path[<->,bend left] (a) edge (c);
\end{tikzpicture}
}%
\\
{%
\begin{tikzpicture}
\node (a) at (0,0) {$a$};
\node (b) at (2,0) {$b$};
\node (c) at (4,0) {$c$};
\path[dashed,->] (a) edge (b);
\path[dashed,->] (c) edge (b);
\end{tikzpicture}
~$\stackrel {ab} \imp$~
\begin{tikzpicture}
\node (a) at (0,0) {$a$};
\node (b) at (2,0) {$b$};
\node (c) at (4,0) {$c$};
\path[<->] (a) edge (b);
\path[dashed,->] (c) edge (b);
\end{tikzpicture}
~$\stackrel {cb} \imp$~
\begin{tikzpicture}
\node (a) at (0,0) {$a$};
\node (b) at (2,0) {$b$};
\node (c) at (4,0) {$c$};
\path[<->] (a) edge (b);
\path[<->] (c) edge (b);
\path[->,bend right] (c) edge (a);
\end{tikzpicture}
}
\end{tabular}
}
\caption{Different executions of the protocol $\LNS$ in distributed dynamic gossip}
\label{fig.intro}
\end{figure}

Usually, the connections in a gossip graph represent what numbers are known by the agents. The connections are pairs in the number relation. We will also model the secrets held by the agents as a binary relation in the gossip graph. On the assumption that you cannot learn someone's secret without knowing her number, the secret relation is contained in the number relation. In the figures in this contribution we visualize the number relation as dashed arrows and the secret relation as solid arrows, we assume reflexivity of relations, and because the secret relation is contained in the number relation we further assume that a solid arrow implies a dashed arrow. Figure \ref{fig.intro}, top row, depicts in this way the execution of the protocol $\LNS$ for three agents $a,b,c$ in a complete graph (for the number relation).

After the call $ab$, $a$ and $b$ know each other's secrets. After the call $bc$, $b$ and $c$ know all secrets. Although $a$ knows $c$'s number, it does not know $c$'s secret; whereas $c$ knows $a$'s secret: therefore, the arrow for pair $(a,c)$ is dashed and the arrow for pair $(c,a)$ is solid. After the final call $ac$ all agents know all secrets. We have achieved success. On this graph, and in fact on all complete graphs, every execution of $\LNS$ will terminate successfully. On the circle graph mentioned above every execution of $\LNS$ will also terminate successfully (note that the second call sequence of length 6 is an execution of $\LNS$). We say that on such graphs $\LNS$ is {\em strongly successful}.

The second row of Figure \ref{fig.intro} depicts the execution of protocol $\LNS$ on a graph with an incomplete number relation. The call $bc$ can be made because $b$ knows $c$'s number but doesn't know $c$'s secret. After that, $b$ and $c$ share their numbers and secrets. Subsequently the call $ab$ can be made, after which $a$ and $b$ know all numbers and all secrets. As they know all secrets, they will make no further calls. Unfortunately, agent $c$ does not know $a$'s secret but also does not know $a$'s number, so cannot make a call. We are stuck. This execution of protocol $\LNS$ does not terminate (i.e., the execution halts without exiting the until-loop of the algorithm, that would result in termination). The third row of Figure \ref{fig.intro} depicts a successful execution $ab;bc;ac$ of $\LNS$ on the same graph. On this graph some executions of $\LNS$ is successful and other executions of $\LNS$ are unsuccessful. We say that on this graph $\LNS$ is {\em weakly successful}.

Finally, the last row of Figure \ref{fig.intro} depicts the execution of protocol $\LNS$ on yet another graph with an incomplete number relation. The extension of $\LNS$ on this graph consists of two executions, $ab;cd$, as depicted, and $cd;ab$. Both get stuck. On this graph, no execution of $\LNS$ is successful.

Now consider the protocol $\SFS$ wherein any call can be made to an agent whose number you know. If we extend the call sequence in the last row of Figure \ref{fig.intro} with call $ab$, then all agents know all secrets. Sequence $ab;cb;ab$ is a successful execution of $\SFS$. Not all $\SFS$ sequences are successful. For example, the infinite sequence $ab;ab;\dots$ is unsuccessful. This is an unfairly scheduled $\SFS$ sequence. But all {\em fair} executions of $\SFS$ are successful on this graph. We say that on this graph $\SFS$ is {\em fairly successful}.

\paragraph*{Results}
We define $6$ dynamic distributed gossip protocols for which we characterize the gossip graphs on which their executions always, or sometimes, or never terminate.  By {\em characterization} we mean that given a gossip protocol and a gossip graph with a certain topology, on any such graph the protocol satisfies the termination condition, and any graph satisfying the termination condition must have that topology. For most of these protocols, including $\SFS$, and possibly assuming fair scheduling, their executions always terminate on weakly connected gossip graphs, and any graph on which their executions always terminate must be weakly connected. If a gossip graph is not weakly connected the protocol will never terminate. The $\LNS$ protocol, wherein an agent may only call another agent if it does not know that agent's secret, is strongly successful on sun graphs, and is weakly successful on any weakly connected gossip graph that is not a bush or a double bush. Both are characterizations. The definitions of sun, bush, double bush will be given in the paper. For now, we note that: every strongly connected graph is a sun graph, that every sun graph is weakly connected and is not a bush or a double bush, and that these inclusions are strict.

\paragraph*{Outline} 
In Section \ref{section.gossip} we define gossip graphs, calls, gossip protocols, strong, weak and fair success, and also the six particular gossip protocols that we characterize in our contribution. The subsequent two sections present those characterization results. Section \ref{sec.all} is devoted to all gossip protocols for which strong and weak success is characterized by weakly connected graphs. Section \ref{section.lns} is devoted to the $\LNS$ protocol, for which strong success and weak success are characterized by non-trivial classes of gossip graphs. Section \ref{section.other} summarizes our results and suggests further research.

\section{Gossip graphs and gossip protocols} \label{section.gossip}

\subsection{Gossip graphs}

Given a finite set of agents (or nodes) $\agents = \{a,b,\ldots\}$ (all lower case letters, possibly quoted or indexed, denote agents), we represent a gossip graph $G$ with numbers and secrets as a triple $(\agents,N,S)$ with $N, S\subseteq \agents^2$. That is, the agents $\agents$ are the vertices and $N, S$ are binary relations on $\agents$, with $Nxy$ (for $(x,y) \in N$) expressing that $x$ knows the (telephone) number of $y$, and $Sxy$ expressing that $x$ knows the secret of $y$.

Let us first introduce standard graph terminology, given a carrier set $A$ and binary relations like $N$ and $S$. We let $\identity_{\agents} = \{(x,x)\}_{x\in\agents}$ be the {\em identity} relation on $\agents$, and {\em converse} relation $N^{-1} = \{ (x,y) \mid Nyx \}$. We write $N_x$ for $\{ y \in \agents \mid Nxy \}$. (We may further write $\neg Nxy$ for $(x,y)\notin N$ and anyway write $xy$ for a pair $(x,y)$, such as for the calls defined below.) Relation $N \circ S = \{ (x,y) \mid \ \text{there is a } z \text{ such that } Nxz \text{ and } Szy \}$ is the {\em composition} of $N$ and $S$, and using that we define $N^1 = N$, $N^{i+1} = N^i \circ N$, and $N^* = \Union_{i \geq 1} N^i$. Relation $N$ is \emph{complete} iff $N = \agents^{2}$, it is \emph{weakly connected} if for all $x,y \in \agents$ there is an $(N \cup N^{-1})$-path from $x$ to $y$, and it is \emph{strongly connected} if for all $x,y \in \agents$ there is an $N$-path from $x$ to $y$. For any set $X$ (like $\agents$, or $N$) we write $X+x$ for $X \union \{x\}$ and $X-x$ for $X \setminus \{x\}$. A pair $xy \in N$ is a \emph{bridge} iff $N-xy$ is not weakly connected. 
\begin{definition}[Gossip graph]\label{gossip-graph}
A \emph{gossip graph} is a triple $G = (\agents, N, S)$ with $N \subseteq \agents^2$ and $S \subseteq \agents^2$. An \emph{initial gossip graph} is a gossip graph with $S= I_{\agents} \subseteq N $. Agent $x$ is an {\em expert} if $S_x = \agents$. An agent or node is {\em terminal} iff $N_x = \{x\}$. Gossip graph $G$ is {\em complete}, {\em weakly (strongly) connected} if $N$ is, respectively, complete, weakly (strongly) connected. For $y \notin A$, $G+y = (\agents+y,N+yy,S+yy)$. For $x,y \in \agents$, $G+xy = (\agents,N+xy,S)$. For $B \subseteq A$, the {\em restriction} of $G$ to $B$ is $G|B = (B, N \inter B^2, S \inter B^2)$.
\end{definition}
In an initial gossip graph each agent only knows its own secret and at least knows its own number. When we employ common graph terminology when referring to a gossip graph, this applies to the number relation, not to the secret relation.

\subsection{Gossip Protocols}

A {\em call} from $x$ to $y$ is a pair $(x,y)$ for which we write $xy$. The call $xy$ in $G$ merges the secrets and the numbers of $x$ and $y$. By $\overline{xy}$ we mean the call that can be either $xy$ or $yx$.

\begin{definition}[Call]\label{def-call}
Let $G = (\agents,N,S)$ and $x,y \in \agents$, and $Nxy$. The \emph{call} $xy$ maps $G$ to gossip graph $G^{xy} = (\agents, N^{xy}, S^{xy})$ defined by
\begin{center}
$N^{xy}_{z} = \begin{cases}N_x \cup N_y & \mbox{if } z \in \{x,y\} \\ 
N_z & \mbox{otherwise}\end{cases}$ 
\qquad and \qquad $S^{xy}_{z} = \begin{cases}S_x \cup S_y & \mbox{if } z \in \{x,y\}\\
S_z & \mbox{otherwise}\end{cases}$
\end{center}
\end{definition}
From this we immediately get that $N^{xy} = N \union (\{(x,y),(y,x)\} \circ N)$ and $S^{xy} = S \union (\{(x,y),(y,x)\} \circ S)$.
\begin{definition}[Call sequence] \label{def.callsequence}
A {\em call sequence} $\sigma$ is a finite or infinite sequence of calls. The empty sequence is $\epsilon$. We write $\sigma;\tau$ for the concatenation of (finite) sequence $\sigma$ and sequence $\tau$. We denote by $\sigma\sqsubseteq\tau$ that $\sigma$ is a prefix of $\tau$ (where $\sigma\sqsubset \tau$ means proper prefix). We use $|\sigma|$ for the length of a sequence. Given a finite sequence $\sigma$ of length $n$, and $i \leq n$, $\sigma[i]$ (for $i\geq 1$) is the $i$th call in $\sigma$, and $\sigma|i$ is the prefix of $\sigma$ consisting of the first $i$ calls. For $x \in \agents$, $\sigma_x$ is the subsequence of calls containing $x$, defined as: $\epsilon_x = \epsilon$, $(\sigma;uv)_x = \sigma_x;uv$ if $x =u$ or $x =v$, and $(\sigma;uv)_x =\sigma_x$ otherwise.

A call $xy$ is {\em possible} given a gossip graph $G = (\agents,N,S)$ if $Nxy$. Call sequence $\epsilon$ is possible for any gossip graph. Call sequence $xy;\sigma$ is possible on $G$ if call $xy$ is possible on $G$ and sequence $\sigma$ is possible on $G^{xy}$. If call sequence $\sigma$ is possible for gossip graph $G$, the ({\em induced}) gossip graph $G^\sigma$ is defined as: $G^\epsilon = G$, $G^{xy;\sigma} = (G^{xy})^{\sigma}$. A {\em call sequence $\sigma$ for $B$} only involves calls between agents in $B \subseteq \agents$.
\end{definition}
Strictly, the gossip graph $G^\sigma$ is different from the pair $(G,\sigma)$ consisting of a gossip graph and a possible call sequence: the history $\sigma$ cannot be retrieved from $G^\sigma$. But informally, we will identify the one with the other. In some gossip protocols properties of prior calls play a role.


We now come to the definition of {\em gossip protocol}. A gossip protocol is a program or procedure for selecting calls for execution that satisfy a {\em protocol condition}. 
\begin{definition}[Protocol condition]\label{def-protocolcondition}
Let an initial gossip graph $G = (\agents,N,S)$ and a possible call sequence $\sigma$ be given. A {\em protocol condition} is a property $\pi(x,y)$ that is a boolean combination of constituents $S^\sigma xy$, $xy \in \sigma_x$, $yx \in \sigma_x$, $\sigma_x = \epsilon$, $\sigma_x = \tau;xz$, and $\sigma_x = \tau;zx$.
\end{definition}
This definition of protocol condition is sufficient to define the protocols treated in this paper. Protocol conditions can be far more general. Informally, $\pi(x,y)$ may be any first-order definable property with (possibly) free variables $x$ and $y$ that satisfies {\em locality}, i.e., such that agent $x$ can select agent $y$ based on the local state of $x$, in other words: based on what $x$ knows. The protocol conditions used in this work satisfy locality in an obvious way. Gossip protocol specification languages are investigated in \cite{apt-tark,DitmarschEPRS17}. 

\begin{definition}[Gossip protocol]\label{def-protocol}
A \emph{gossip protocol} $\protocol$ with protocol condition $\pi(x,y)$ is a non-deterministic algorithm of the following shape, operating on any given initial gossip graph $G = (\agents,N,S)$: \begin{quote} Until all agents are experts, select $x,y\in\agents$ such that $x \neq y$, $Nxy$, and $\pi(x,y)$, and execute call $xy$. \end{quote}
\end{definition}
In the introduction we explained how this can be seen as defining a distributed protocol. An alternative formulation that avoids abnormal termination (getting stuck) is:
\begin{quote} {\em Until all agents are experts and there are $x,y\in\agents$ such that $x \neq y$, $Nxy$, and $\pi(x,y)$, select $x,y\in\agents$ such that $x\neq y$, $Nxy$, and $\pi(x,y)$, and execute call $xy$.} \end{quote}
\begin{definition}[Permitted call sequence]\label{def-permitted}
Let a \emph{gossip protocol} $\protocol$ with protocol condition $\pi(x,y)$ be given. Let $G = (\agents,N,S)$ be a gossip graph.
\begin{itemize}
\item call $xy$ is $\protocol$-permitted on $G^\sigma$ iff $\sigma$ is possible on $G$, $x\neq y$, $N^\sigma{xy}$, and $\pi(x,y)$ holds in $G^\sigma$. 
\item call sequence $\epsilon$ is $\protocol$-permitted on $G$.
\item call sequence $\sigma;xy$ is $\protocol$-permitted on $G$ iff $\sigma$ is $\protocol$-permitted on $G$ and $xy$ is $\protocol$-permitted on $G^\sigma$. 
\item an infinite call sequence $\sigma$ is $\protocol$-permitted on $G$ iff for all $n \in \Nat$, $\sigma[n+1]$ is $\protocol$-permitted on $G^{\sigma|n}$. \end{itemize}
\end{definition}
A $\protocol$-permitted call sequence is also called a {\em $\protocol$-sequence}.

\begin{definition}[Protocol extension]\label{def-protocolextension}
The {\em extension} ${\extension \protocol G}$ of protocol $\protocol$ on $G$ is the set of $\protocol$-permitted sequences on $G$. If all call sequences in $\extension \protocol G$ are finite (i.e., terminating), protocol $\protocol$ is {\em terminating} on $G$. Given protocols $\protocol$ and $\protocol'$ and a collection ${\mathcal G}$ of gossip graphs, we write $\protocol_{\mathcal G} \subseteq \protocol'_{\mathcal G}$ iff $\extension \protocol G \subseteq \extension {\protocol'} G$ for all $G \in {\mathcal G}$, and $\protocol \subseteq \protocol'$ if $\protocol_{\mathcal G} \subseteq \protocol'_{\mathcal G}$ holds for the collection ${\mathcal G}$ of initial gossip graphs.
\end{definition}

\begin{definition}[Maximal, stuck, fair]\label{def-maximalstuckfair} Let $G = (A,N,S)$ and $\protocol$ be given.
\begin{itemize}
\item A {\em $\protocol$-maximal} sequence $\sigma$ on $G$ is a sequence $\sigma$ that is $\protocol$-permitted on $G$ and that is infinite or such that no call is $\protocol$-permitted on $G^\sigma$. 
\item A finite call sequence $\sigma$ is {\em $\protocol$-stuck} on $G$ iff $S^{\sigma}$ is not complete and $\sigma$ is $\protocol$-maximal on $G$.
\item Call sequence $\sigma\in\protocol$ is {\em fair} on $G$ iff $\sigma$ is finite or, whenever $\sigma$ is infinite, then for all $xy$: if for all $i \in \Nat$ there is a $j \geq i$ such that $xy$ is $\protocol$-permitted on $G^{\sigma|j}$, then for all $i \in \Nat$ there is a $j \geq i$ such that $xy = \sigma[j]$. 
\end{itemize}
\end{definition}
Given $B \subseteq \agents$, similarly to $\protocol$-maximal we can define that a sequence $\sigma$ is {\em $\protocol$-maximal for $B$}, namely if all calls in $\sigma$ are between members of $B$. If a call sequence is stuck then no further calls can be made but some agents do not know all secrets. If the context makes clear what the protocol $\protocol$ is, instead of $\protocol$-maximal, $\protocol$-stuck, $\protocol$-fair, and $\protocol$-permitted we write maximal, stuck, fair, and permitted. 

\begin{definition}[Success] \label{def-successfulterminating}
Let a gossip graph $G$ and a protocol $\protocol$ be given. A call sequence $\sigma\in\extension\protocol G$ is {\em successful} (or {\em $\protocol$-successful}) if it is finite and if in $G^\sigma$ all agents are experts. 
\begin{itemize}
\item Protocol $\protocol$ is {\em strongly successful} on $G$ if all maximal $\sigma \in \extension \protocol G$ are successful. 
\item Protocol $\protocol$ is {\em fairly successful} on $G$ if all maximal fair $\sigma \in  \extension \protocol G$ are successful.
\item Protocol $\protocol$ is {\em weakly successful} on $G$ if there is a $\sigma\in\extension \protocol G$ that is successful.
\item Protocol $\protocol$ is {\em unsuccessful} on $G$ if there is no $\sigma\in\extension \protocol G$ that is successful.
\end{itemize}
Given a collection ${\mathcal G}$ of gossip graphs, $\protocol$ is (strongly, fairly, weakly, un-) successful on ${\mathcal G}$ iff $\protocol$ is (strongly, fairly, weakly, un-) successful on every $G \in {\mathcal G}$.
\end{definition}

A finite call sequence is fair by definition. If a protocol is fairly successful, then all fair sequences in the extension are finite. Strongly successful implies weakly successful, as the set of maximal call sequences is non-empty (if nothing is permitted, the empty sequence is maximal). If a sequence is successful, it is also maximal. Strongly successful also implies fairly successful, as a successful call sequence is finite and thus fair. Fairly successful implies weakly successful, as $\epsilon$ is fair and each fair call sequence can be extended into a maximal fair call sequence, and therefore the set of maximal fair call sequences is non-empty. So strong implies fair and fair implies weak. Unsuccessful is the same as not weakly successful.

\begin{definition}[Gossip problem] \label{def.gossipproblem}
Given a collection ${\mathcal G}$ of gossip graphs and a protocol $\protocol$, the {\em gossip problem} is: is $\protocol$ (strongly, fairly, weakly, un-) successful on ${\mathcal G}$?
\end{definition}

\subsection{Elementary combinatorial results} \label{sec.elem}

We close this section with some elementary combinatorial results on gossip graphs.

\begin{lemma}\label{Coro:AccSecr}\label{Thm:Incl}
Let $G = (\agents,N,S)$ be an initial gossip graph, and let $\sigma$ be a possible call sequence for $G$. Then: \begin{enumerate} \item $S^{\sigma}\subseteq N^{\sigma}$ \item  $S^\sigma \circ N \subseteq N^\sigma$ \end{enumerate} \end{lemma}
\begin{proof} Both are proved by induction on $\sigma$. 
\begin{enumerate} \item Initially, $S_x \subseteq N_x$. Then, it follows from $S^\sigma_x \subseteq N^\sigma_x$ and $S^\sigma_y \subseteq N^\sigma_y$ that $S^\sigma_x \cup S^\sigma_y \subseteq N^\sigma_x \cup N^\sigma_y$, and therefore $S^{\sigma;xy}_x = S^{\sigma;xy}_y \subseteq N^{\sigma;xy}_x = N^{\sigma;xy}_y$.
\item 
For the base case we have that $S \circ N = I_\agents \circ N = N$. 

For the induction step, assume $S^\sigma \circ N \subseteq N^\sigma$, let $xy$ be a possible call in $G^\sigma$, and let $(S^{\sigma;xy} \circ N)ab$. If $(S^\sigma \circ N){ab}$, then by the induction hypothesis, $N^\sigma{ab}$, and hence by $N^\sigma \subseteq N^{\sigma;xy}$ we get that $N^{\sigma;xy}ab$. 

If $\neg(S^\sigma \circ N)ab$, then we may assume (w.l.o.g.) that $a = x$ and that there is some $z$ with $S^{\sigma;xy} xz$ and $Nzb$. From $S^{\sigma;xy}xz$ it follows that either $S^\sigma xz$ or $S^\sigma yz$. In the former case, we have $(S^\sigma \circ N)xb$, and therefore by the induction hypothesis, $N^\sigma{xb}$. In the latter case, we have $(S^\sigma \circ N)yb$, and therefore by the induction hypothesis, $N^\sigma{yb}$. From $N^\sigma{xb}$ or $N^\sigma{yb}$ it follows by the definition of $N^{\sigma;xy}$ that $N^{\sigma;xy}{xb}$. 
\end{enumerate} \vspace{-.7cm}
\end{proof}
Therefore, if we begin in a situation where we know more numbers than secrets, we cannot learn all secrets without learning all numbers.

It will be obvious that, if $\sigma$ is possible for a gossip graph $G = (\agents,N,S)$, then $G$ is weakly connected iff $G^\sigma$ is weakly connected: after any possible call, the number of connected components is invariant. A fortiori, both the number relation and the secret relation are then incomplete on $G^\sigma$. Thus, the goal of all agents being expert can never be reached in gossip graphs that are not weakly connected. This begs the question what the minimum structural requirements are for protocols to be successful. We will answer this question for the protocols defined in the following section.

\subsection{Gossip protocols in this contribution}

The following protocols are investigated in this contribution ($\top$ is the trivial proposition).

\begin{definition}[Gossip Protocols] We define the following protocols, where we give their  name followed by the protocol condition and on the next line an informal description.
\begin{itemize}
\item $\SFS$ \quad\quad $\top$ \\
Until every agent knows all secrets, choose different agents $x$ and $y$ such that $x$ knows the number of $y$, and let $x$ call $y$.
\item $\TOK$ \quad\quad  $\sigma_x = \epsilon \lor \sigma_x = \tau;zx$ \\ 
Until every agent knows all secrets, choose different agents $x$ and $y$ with $x \neq y$ such that $x$ knows $y$'s number and either $x$ has not been in prior calls or the last call involving $x$ was {\bf to} $x$, and let $x$ call $y$.
\item $\SPI$ \quad\quad  $\sigma_x = \epsilon \lor \sigma_x = \tau;xz$ \\ 
Until every agent knows all secrets, choose different agents $x$ and $y$ such that $x$ knows $y$'s number and either $x$ has not been in prior calls or the last call involving $x$ was {\bf from} $x$, and let $x$ call $y$.
\item $\CO$ \quad\quad  $xy \notin\sigma_{x} \land yx \notin\sigma_{x}$ \\
Until every agent knows all secrets, choose different agents $x$ and $y$ such that $x$ knows the number of $y$ and there was no prior call between $x$ and $y$, and let $x$ call $y$.
\item $\wCO$ \quad\quad  $xy \notin\sigma_{x}$ \\
Until every agent knows all secrets, choose different agents $x$ and $y$ such that $x$ knows the number of $y$ and $x$ did not call $y$ before, and let $x$ call $y$.
\item $\LNS$ \quad\quad  $\neg S^\sigma xy$ \\
Until every agent knows all secrets, choose different agents $x$ and $y$ such that $x$ knows the number of $y$ but not the secret of $y$, and let $x$ call $y$. 
\end{itemize}
\end{definition}
We recall that in all our protocols both secrets and numbers are exchanged. We now succinctly  describe each protocol and refer to prior descriptions in the non-dynamic gossip literature.

The $\SFS$ protocol ($\SFS$ stands for {\em Any} call) is like the random selection of a call that is usual for gossip protocols in the networks community \cite{eugster}, except that we have the dynamic interpretation as in \cite{rajaraman,eugster}.

In the $\TOK$ protocol ($\TOK$ stands for {\em Token}), an agent must have a token in order to make a call. Initially, each agent has a token. However, if $x$ calls $y$, $x$ hands over her token to $y$, so she will have to wait until she is called by another agent before she can call again. If $y$ already had a token, he will keep the token, in other words, the two tokens then merge into one. In the $\SPI$ protocol ($\SPI$ for {\em Spider}), in contrast, when $x$ calls $y$, $y$ hands over his token to $x$. Agent $y$ can never again make a call. Either way, the number of agents holding a token is weakly decreasing. We do not know of occurrences of the $\TOK$ and $\SPI$ protocol conditions in the gossip literature, but it seems very likely that such works exist. We were motivated in their formulation by the expectation that such protocols would have faster expected termination when compared to making random calls ($\SFS$). For example, a consequence of $\TOK$ is that this rules out immediate (`useless') subsequent calls to the same agent, as in $ab;ab;\dots$, but in subsequent work this has been disproved \cite{DitmarschKS17}.

In the $\CO$ protocol ($\CO$ for {\em Call Once}), agents $x$ and $y$ can call each other only once. It does not matter who initiated the previous call. All $\CO$-permitted call sequences are finite, so fairness is not an issue. Protocol $\wCO$ (for {\em Weak Call Once}) is a variant of $\CO$ that only requires the same call not to have been made. Protocol $\wCO$ is reminiscent of the {\em quasi-random} protocols in \cite{doerr}: instead of choosing to call any neighbour (connected node) in any round, in the next round the neighbours of a node are selected in a fixed, cyclical order: ``{\em Contacts are chosen uniformly at random among all neighbors except the one that was chosen just in the round before.} \cite[p.\ 9:2]{doerr}.''

In the logic community the $\LNS$ protocol ($\LNS$ for {\em Learn New Secrets}) has been investigated for complete graphs in \cite{maduka-ecai,Attamah2017,apt-tark}, and in the networks community in, for example, \cite[Algorithm 3]{Haeupler15}, however, for rounds of calls and not for call sequences, which results in different secrets distributions and protocol extensions. It is interesting to compare the $\LNS$ protocol condition with the condition NOHO, `No One Hears Own', pioneered in \cite{BAKER1972191}. NOHO corresponds to $\neg S^\sigma yx$ for call $xy$. Interestingly, in \cite{hedetniemietal:1988} NOHO is described as ``{\em no one can call a person if the caller already knows the unique piece of information originally only known by the person called},'' which, we think, amounts to the $\LNS$ protocol condition $\neg S^\sigma xy$ for call $xy$; the other referenced works in \cite{hedetniemietal:1988} all use the \cite{BAKER1972191} interpretation. The NOHO condition is hard to enforce in distributed gossip: after sequence $ab;bc$, this would rule out subsequent call $ac$ as $c$ already knows the secret of $a$; it learnt that in call $bc$. But agent $a$ cannot know that. The NOHO condition is not local.

To give the reader an idea of the differences between these protocols we compare their extensions. We recall that $\protocol \subseteq \protocol'$ stands for $\protocol_\mathcal{G} \subseteq \protocol'_\mathcal{G}$ for the class $\mathcal{G}$ of initial gossip graphs. These protocol extensions determine a partial order on $\mathcal{G}$. We have that ${\LNS} \subset {\CO} \subset {\wCO} \subset {\SFS}$, and that $\SPI$ and $\TOK$ are incomparable and both of them are also incomparable to any of $\LNS$, $\CO$, and $\wCO$. These results are depicted in Figure \ref{fig.hierarchy}. 

\begin{figure}
\begin{center}
\fbox{
\begin{tikzpicture}
\draw (0.5, 1.5) ellipse (1.8 and 1.4);
\node (a2) at (-4,0) {\footnotesize $a$};
\node (b2) at (-3,0) {\footnotesize $b$};
\node (c2) at (-2,0) {\footnotesize $c$};
\path[->,dashed] (a2) edge (b2);
\path[->,dashed] (b2) edge (c2);
\node (seq2) at (0.43, 1) {\footnotesize \begin{turn}{27} $bc;ab;ac$\end{turn}};	
\node (seq2bisbis) at (0, 2) {\footnotesize \begin{turn}{27} $ab;ac;ab$\end{turn}};	
\node (spi) at (-0.2, 2.5) {\footnotesize \begin{turn}{24}$\SPI$\end{turn}};	

\draw (3, 1.5) ellipse (1.8 and 1.4);
\node (seq4) at (5.7, 0) {\footnotesize $ab;ab;ba$};

\node (seq1) at (2.7, .55) {\footnotesize \begin{turn}{-29}$ab;bc$\end{turn}};	
\node (tok) at (3.85, 2.45) {\footnotesize \begin{turn}{-29}$\TOK$\end{turn}};	

\node (seq2bis) at (3.1, 1) {\footnotesize \begin{turn}{-29} $ab;bc;ca$\end{turn}};	
\node (seq3bis) at (3.4, 1.5) {\footnotesize \begin{turn}{-29} $ab;ba$\end{turn}};	
\node (seq4bis) at (3.7, 1.95) {\footnotesize \begin{turn}{-29} $ab;ba;ab$\end{turn}};

\draw (1.75, -.2) ellipse (1.8 and 1.4);

\node (seq6) at (.8, 0.55) {\footnotesize \begin{turn}{27} $ab;ac$\end{turn}};	

\node (seq7) at (1.75, .9) {\footnotesize $ab$};	

\node (emptyseq0) at (1.75, 2.3) {\footnotesize --};
\node (emptyseq1) at (1.75, 1.9) {\footnotesize ---};
\node (emptyseq2) at (1.75, 1.4) {\footnotesize ---};
\node (emptyseq3) at (.3, 1.6) {\footnotesize \begin{turn}{27} ---\end{turn}};

\node (seq8) at (1.75, -.4) {\footnotesize $ab;bc;ac$};	
\node (lns) at (2.7, -1.05) {\footnotesize \begin{turn}{27}$\LNS$\end{turn}};	
\node (ss) at (7, 2.7) {\footnotesize $\SFS$};	

\draw (1.75, -.2) ellipse (2.4 and 1.9);
\node (co) at (3, -1.52) {\begin{turn}{27} \footnotesize $\CO$\end{turn}};
\node (seq9) at (1.75, -1.82) {\footnotesize $ab;ac;cb$};

\draw (1.75, -.2) ellipse (3 and 2.4);
\node (wco) at (3.25, -2) {\begin{turn}{27}\footnotesize $\wCO$\end{turn}};
\node (seq10) at (1.75, -2.35) {\footnotesize $ab;bc;ba$};

\end{tikzpicture}
}
\caption{The protocol extension hierarchy, illustrated for the gossip graph on the left}
\label{fig.hierarchy}
\end{center}
\end{figure}
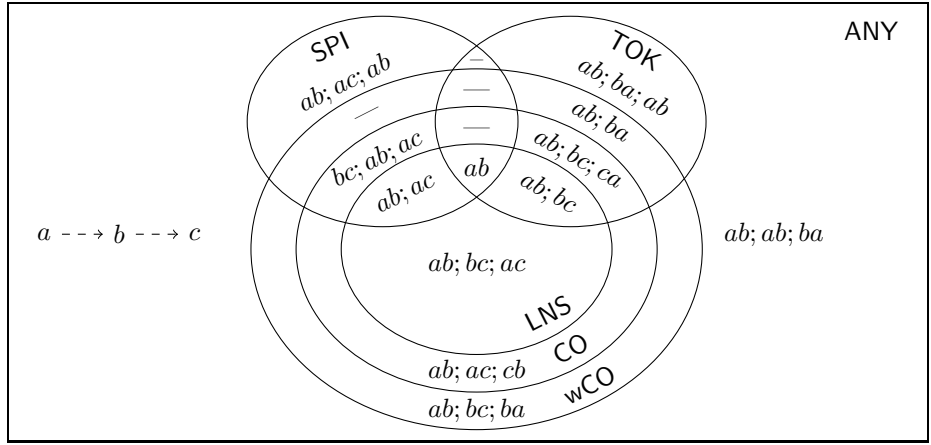

Each area in Figure \ref{fig.hierarchy} represents the set of call sequences for all gossip graphs, that are determined in the obvious way by algebraic manipulation of protocol extensions. As representatives of these sets we have chosen call sequences for the initial gossip graph consisting of three agents $a,b,c$ such that $Nab$ and $Nbc$. For example, the area containing $ab;ba$ defines $\TOK \inter \wCO \inter \overline{\CO} \inter \overline{\SPI}$. This means that $ab;ba$ is also an $\SFS$ call sequence, but not a $\CO$ (and therefore also not a $\LNS$) or $\SPI$ call sequence (call $ba$ is then not permitted after $ab$). The sequence $ab;ab;ba$ is an $\SFS$ call sequence that is not permitted in any other of the protocols defined in the paper. If an area contains `---', then the corresponding extension is empty, and not only for the example gossip graph, but for all gossip graphs. For example, $\overline{\LNS} \inter \SPI \inter \TOK = \emptyset$. Indeed, a call sequence $\sigma$ that is not in $\LNS$ cannot be both in $\SPI$ and $\TOK$: let $xz$ be the first not $\LNS$-permitted call in $\sigma$, there must then be a previous call involving $x$ in $\sigma$; if that call was from $x$, then $xz$ is not $\TOK$-permitted, but if the call was to $x$, then $xz$ is not $\SPI$-permitted. 

\section{Protocols successful on weakly connected graphs} \label{sec.all}

For more than two agents, Protocol $\SFS$ is not strongly successful on initial gossip graphs. Suppose there are at least three agents $a,b,c$ and let call $ab$ be possible, i.e., $\SFS$-permitted. Then the infinite sequence $ab; ab; ab; \ldots$ is also $\SFS$-permitted and agent $c$ will never learn the secrets of $a$ and $b$. As $ab; ab; ab; \ldots$ is also $\SPI$-permitted, and as $ab; ba; ab; ba; \ldots$ is $\TOK$-permitted, in all three cases we can only wish for fair success.

\begin{theorem} \label{theo.all} \label{Thm-SFS} \label{theo.tok} \label{SPI1} 
Let $G$ be an initial gossip graph.
\begin{enumerate}
\item Protocol $\SFS$ is fairly successful on $G$ iff $G$ is weakly connected.
\item Protocol $\TOK$ is fairly successful on $G$ iff $G$ is weakly connected.
\item Protocol $\SPI$ is fairly successful on $G$ iff $G$ is weakly connected.
\end{enumerate}
\end{theorem}

\begin{proof} In all cases we only show the direction from right to left of the equivalence, as the direction from left to right is trivial. In the proof we use the following properties. \begin{quote} {\em If $G$ is weakly connected and $\sigma$ is a call sequence such that for all $x,y \in A$ with $Nxy$ we have $S^\sigma_x = S^\sigma_y$, then $\sigma$ is successful.} \hfill $(i)$ \end{quote} This follows from the weak connectivity of $G$, and the fact that initially each agent knows its own secret. \begin{quote} {\em A fair infinite call sequence $\sigma$ has a finite prefix $\tau$ after which the graph is stable, i.e., the numbers and secrets (and, in case of $\SPI$ and $\TOK$, the number of tokens) do not change in the continuation.} \hfill $(ii)$ \end{quote} That follows from the simple observation that the numbers and secrets are weakly increasing after each call and bounded (in the case of tokens, it follows because they are initially bounded, weakly decreasing in each call, and because there is at least one).
\begin{quote} {\em An invariant under $\TOK$ and $\SPI$ execution is that for any agent, its number is known by an agent holding a token} \hfill $(iii)$ \end{quote} Initially this is trivially true as all agents hold a token, and it is clear that whenever two agents call this property is preserved. 

We continue with the actual proof. Let $G = (\agents,N,S)$ be weakly connected. Assume towards a contradiction that there is an infinite fair call sequence $\sigma$ on $G$. Let $\tau$ be a stable prefix of $\sigma$. Because of $(i)$, there must be $x,y \in A$ such that $Nxy$ and $S^\tau_x \neq S^\tau_y$. 

\begin{enumerate}

\item Let $\sigma$ be $\SFS$-permitted. Since $xy$ is $\SFS$-permitted and $\sigma$ is fair, call $xy$ will eventually happen in $\sigma$ after $\tau$. Contradiction with $(ii)$.

\item

Let $\sigma$ be $\TOK$-permitted. Let $z$ be an agent holding a token who knows the number of $x$ $(iii)$. If $z=x$, then $xy$ is $\TOK$-permitted on $G^{\tau}$. Otherwise, $zx$ is $\TOK$-permitted on $G^\tau$ and $xy$ is $\TOK$-permitted in $G^{\tau;zx}$. Call $xy$ will therefore eventually happen in $\sigma$ after $\tau$. Contradiction with $(ii)$.

\item 

Let $\sigma$ be $\SPI$-permitted. Using $(iii)$, let token holder $z$ know the number of $x$ and let token holder $w$ know the number of $y$. As calls $zx$, $wy$, $zx$ are successively $\SPI$-permitted in $\sigma$ after $\tau$, eventually $x$ and $y$ will have the same set of secrets. Contradiction with $(ii)$.
\end{enumerate} \vspace{-.7cm}
\end{proof}

\begin{corollary}
Any fair $\SFS$/$\TOK$/$\SPI$-permitted sequence $\sigma$ is finite.
\end{corollary}

\begin{theorem} \label{calloncetheorem}
Let $G$ be an initial gossip graph.
\begin{enumerate}
\item Protocol $\CO$ is strongly successful on $G$ iff $G$ is weakly connected.
\item Protocol $\wCO$ is strongly successful on $G$ iff $G$ is weakly connected.
\end{enumerate}
\end{theorem}

\begin{proof} Again, we only show the non-trivial direction of the equivalence, from right to left. Let $G = (\agents,N,S)$ be weakly connected.
\begin{enumerate}
\item Let $\sigma$ be a $\CO$-maximal call sequence. As $G$ is weakly connected, there is an undirected path $\pi$ in $G$ between any two agents $a \neq c$. Assume towards a contradiction that $\neg S^\sigma{ca}$. Let $b$ be the first agent on path $\pi$ who does not know the secret of $a$ ($b$ may be $c$). Let $d$ be the predecessor of $b$ on the part of $\pi$ from $a$ to $b$ ($d$ may be $a$). So we have this situation: $(a,\dots,d,b,\dots,c)$. By definition of the path, $Nbd$ or $Ndb$. As $\sigma$ is maximal and not all agents are experts in $G^\sigma$, $bd \in \sigma$ or $db \in \sigma$. Agent $b$ is the first agent for which $\neg S^\sigma ba$, so $S^\sigma da$, and therefore $N^\sigma da$. Again, as $\sigma$ is maximal and not all agents are experts in $G^\sigma$, $da \in \sigma$ or $ad \in \sigma$. If $\overline{ad}$ is before $\overline{db}$ in $\sigma$ (we recall $\overline{xy}$ means `$xy$ or $yx$') then $b$ learns the secret of $a$ from $d$, which contradicts $\neg S^\sigma ba$. But if $\overline{db}$ is before $\overline{ad}$ in $\sigma$, then $a$ knows the number of $b$ after call $\overline{ad}$. We also know that $ab$ and $ba$ are not in $\sigma$, as $b$ does not know $a$'s secret. Therefore, call $\overline{ab}$ is $\CO$-permitted after $\sigma$, which contradicts the maximality of $\sigma$. 
\item This follows from the previous item.
\end{enumerate} \vspace{-.7cm}
\end{proof}

\begin{corollary} \label{cor.errare}
Let $G$ be an initial gossip graph and $\protocol$ one of $\SFS,\TOK,\SPI,\CO,\wCO$. Then $\protocol$ is weakly successful on $G$ iff $G$ is weakly connected.
\end{corollary}
\begin{proof}
If $G$ is weakly connected, then it follows from Theorem~\ref{theo.all} that $\protocol$ is either fairly successful or strongly successful. In both cases it is therefore weakly successful. If $G$ is not weakly connected, then $\protocol$ is trivially unsuccessful, i.e., not weakly successful (see Section \ref{sec.elem}).
\end{proof}

\section{Graph characterization of success for LNS} \label{section.lns}

All $\LNS$-sequences are finite, so fairness is not an issue. In Subsection \ref{subs.1} we show that $\LNS$ is strongly successful on a gossip graph iff that graph is a sun. This characterizes strong success for $\LNS$. In Subsection \ref{subs.2} we show that $\LNS$ is {\bf not} weakly successful on a gossip graph that is a bush or a double bush. In Subsection \ref{subs.3} we show that $\LNS$ is weakly successful on a gossip graph that is {\bf not} a bush or a double bush. Together, this characterizes weak success for $\LNS$.

\subsection{Where $\LNS$ is strongly successful and not strongly successful} \label{subs.1}

Employing some lemmas we characterize the initial gossip graphs for which $\LNS$ is strongly successful. As each agent will make at most $|\agents|-1$ calls, all executions of $\LNS$ are finite. 

\begin{lemma} \label{thm:maxSigmaEq}
If $G = (\agents, N, S)$ is an initial gossip graph and $\sigma$ is $\LNS$-maximal on $G$, then $S^\sigma = N^\sigma$.
\end{lemma}
\begin{proof}
Let $x,y \in A$ with $N^\sigma xy$ and not $S^\sigma xy$. Then the call $xy$ is permitted in $G^\sigma$, which contradicts the maximality of $\sigma$. This shows $N^\sigma \subseteq S^\sigma$. The property $S^\sigma \subseteq N^\sigma$ follows from Proposition \ref{Coro:AccSecr}.
\end{proof}

\begin{lemma} \label{thm:maxStarEq}
If $\sigma$ is $\LNS$-maximal on an initial gossip graph $G$, then $S^\sigma \circ N^* = S^\sigma$. 
\end{lemma}
\begin{proof}
We have that $S^\sigma \subseteq S^\sigma \circ N^*$ by definition of $N^*$. We now prove that $S^\sigma \circ N^* \subseteq S^\sigma$: let $(x,y) \in S^\sigma \circ N^*$. Then for some $k \in \Nat$, $(x,y) \in S^\sigma \circ N^k$. We get from Lemma \ref{Thm:Incl} plus Proposition \ref{thm:maxSigmaEq} that $S^\sigma \circ N \subseteq S^\sigma$. Applying this fact $k$ times yields $(x,y) \in S^\sigma$.
\end{proof}

\begin{definition}[Sun]
An initial gossip graph $G = (\agents,N,S)$ is a {\em sun} iff $N$ is strongly connected on the restriction of $G$ to the set of non-terminal nodes.
\end{definition}
Figure \ref{fig.sungraph} depicts a typical sun graph. 

\begin{figure}
 \begin{center}
\scalebox{.7}{
\begin{tikzpicture}
\node (a) at (1,3.4) {\footnotesize{$\bullet$}};
\node (b) at (2.4,3.4) {\footnotesize{$\bullet$}};
\node (c) at (3.4,2.4) {\footnotesize{$\bullet$}};
\node (d) at (3.4,1) {\footnotesize{$\bullet$}};
\node (e) at (2.4,0) {\footnotesize{$\bullet$}};
\node (f) at (1,0) {\footnotesize{$\bullet$}};
\node (g) at (0,1) {\footnotesize{$\bullet$}};
\node (h) at (0,2.4) {\footnotesize{$\bullet$}};
\node (a1) at (.5,4.4) {\footnotesize{$\bullet$}};
\node (b1) at (2.9,4.4) {\footnotesize{$\bullet$}};
\node (c1) at (4.4,2.9) {\footnotesize{$\bullet$}};
\node (d1) at (4.4,.5) {\footnotesize{$\bullet$}};
\node (e1) at (2.9,-1) {\footnotesize{$\bullet$}};
\node (f1) at (.5,-1) {\footnotesize{$\bullet$}};
\node (g1) at (-1,.5) {\footnotesize{$\bullet$}};
\node (h1) at (-1,2.9) {\footnotesize{$\bullet$}};

\path[dashed,<->] 
(a) edge (b)
(b) edge (c)
(c) edge (d)
(d) edge (e)
(e) edge (f)
(f) edge (g)
(g) edge (h)
(h) edge (a)
;
\path[dashed,->] (a) edge (a1);
\path[dashed,->] (b) edge (b1);
\path[dashed,->] (c) edge (c1);
\path[dashed,->] (d) edge (d1);
\path[dashed,->] (e) edge (e1);
\path[dashed,->] (f) edge (f1);
\path[dashed,->] (g) edge (g1);
\path[dashed,->] (h) edge (h1);
\end{tikzpicture}
}
\end{center}
\caption{A initial gossip graph that is a sun}
\label{fig.sungraph}
\end{figure}
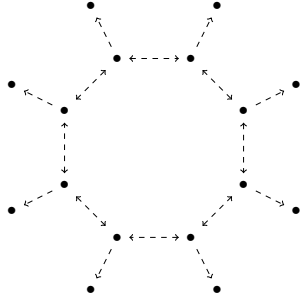

\begin{theorem} \label{theorem.lnsstrong} \label{if_sSC} \label{if_not_sSC}
Let $G$ be an initial gossip graph. Protocol $\LNS$ is strongly successful on $G$ iff $G$ is a sun.
\end{theorem}

\begin{proof} The direction from left to right is proved by contraposition. Let $G = (A,N,S)$ be weakly connected but not a sun. Let a subgraph $H$ of $G$, with carrier set $B$, be a strongly connected component of $G$. Let now $\sigma'$ and $\sigma''$ be $\LNS$-maximal sequences for $\agents\setminus B$ and for $B$, respectively, and let $s(B)$ be the set of agents in $\agents\setminus B$ that are successors of agents in $B$. We distinguish two cases for which $G$ is not a sun:

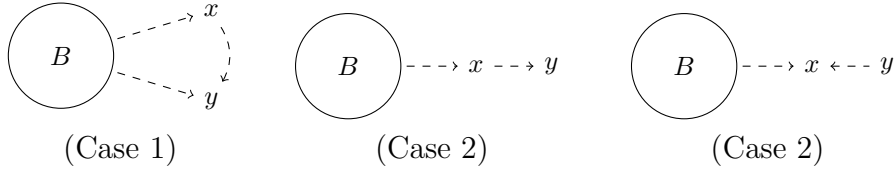
\begin{center}
\begin{tabular}{c@{\qquad}c@{\qquad}c}
\begin{tikzpicture}
\node (B) at (-1, 0) {};
\node (x) at (1,.6) {\footnotesize $x$};
\node (y) at (1,-.6) {\footnotesize $y$};
\path[->,dashed] (B) edge (x);
\path[->,dashed] (B) edge (y);
\path[->,dashed,bend left] (x) edge (y);
\draw (-1, 0) [fill=white] ellipse (.7 and .7);
\node (B2) at (-1, 0) {\footnotesize $B$};
\end{tikzpicture}
&
\begin{tikzpicture}
\node (B) at (-1, 0) {};
\node (x) at (.7,0) {\footnotesize $x$};
\node (y) at (1.7,0) {\footnotesize $y$};
\path[->,dashed] (B) edge (x);
\path[->,dashed] (x) edge (y);
\draw (-1, 0) [fill=white] ellipse (.7 and .7);
\node (B2) at (-1, 0) {\footnotesize $B$};
\end{tikzpicture}
&
\begin{tikzpicture}
\node (B) at (-1, 0) {};
\node (x) at (.7,0) {\footnotesize $x$};
\node (y) at (1.7,0) {\footnotesize $y$};
\path[->,dashed] (B) edge (x);
\path[<-,dashed] (x) edge (y);
\draw (-1, 0) [fill=white] ellipse (.7 and .7);
\node (B2) at (-1, 0) {\footnotesize $B$};
\end{tikzpicture}\\

(Case 1) & (Case 2) & (Case 2)

\end{tabular}
\end{center}

(Case 1) If $s(B) = A\setminus B$, there must be agents $x,y\in s(B)$ such that $Nxy$ (otherwise $G$ is a sun). After $\sigma'$ we have $S^{\sigma'}xy$. Let $Bx$ be a call sequence where everyone in $B$ calls $x$, and let $\sigma'''$ be a maximal call sequence for $B\cup(s(B)\setminus\{y\})$ in graph $G^{\sigma';\sigma''; Bx}$. Then $\sigma';\sigma''; Bx; \sigma'''$ is $\LNS$-maximal on $G$. Clearly agent $y$ is not an expert after this sequence.

(Case 2) There is an agent $y$ in $\agents\setminus B$ who is not the successor of any node in $B$ (see the 2 cases depicted above). Let $\sigma'''$ be a maximal $\LNS$-sequence for $B\cup s(B)$ in $G^{\sigma';\sigma''}$, so that $\sigma';\sigma''; \sigma'''$ is $\LNS$-maximal in $G$. Again, agent $y$ is not an expert.

\medskip

For the other direction, let now $G$ be a sun. Let $\sigma$ be a $\LNS$-maximal call sequence on $G$. Let $x,y \in \agents$. We show that $S^\sigma x y$.

If $x$ is not a terminal, then $N^* xy$. From that and $S^\sigma xx$ follows $(S^\sigma \circ N^*)xy$. From $S^\sigma \circ N^* = S^\sigma$ (Proposition \ref{thm:maxStarEq}) it follows that $S^\sigma xy$. 

If $x$ is a terminal, then by maximality of $\sigma$, there is some $u$ with $ux \in \sigma$. Therefore $N^\sigma xz$ for some $z$ with $Nzx$, and from the maximality of $\sigma$ then follows $S^\sigma xz$. Since $z$ is not terminal, $N^*zy$. From $S^\sigma xz$ and $N^*zy$ we get $(S^\sigma \circ N^*)xy$. By Proposition \ref{thm:maxStarEq} we then get $S^\sigma xy$. 
\end{proof}

\subsection{Where $\LNS$ is not weakly successful} \label{subs.2}

Next on our list is weak success. In the introduction we saw a gossip graph that is not a sun on which some sequence is $\LNS$-successful and another maximal sequence is unsuccessful. Given a graph on which $\LNS$ is unsuccessful, add an edge, or a node and a edge, and a successful sequence may exist yet again. For example, on the gossip graph on the left $\LNS$ is unsuccessful. But adding any edge makes it weakly successful, as evidenced by the successful sequence below it (see also Figure \ref{fig.doublebushex2}, later.)

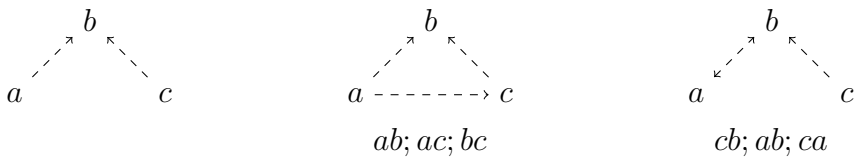
\begin{center}
\begin{tabular}{ccc}
\begin{tikzpicture}
\node (a) at (0,0) {$a$};
\node (b) at (1,1) {$b$};
\node (c) at (2,0) {$c$};
\draw[dashed, ->] (a) -- (b); 
\draw[dashed, ->] (c) -- (b); 
\end{tikzpicture} &
\hspace{1.5cm}
\begin{tikzpicture}
\node (a) at (0,0) {$a$};
\node (b) at (1,1) {$b$};
\node (c) at (2,0) {$c$};
\draw[dashed, ->] (a) -- (b); 
\draw[dashed, ->] (c) -- (b); 
\draw[dashed, ->] (a) -- (c); 
\end{tikzpicture} 
\hspace{1.5cm}
&
\begin{tikzpicture}
\node (a) at (0,0) {$a$};
\node (b) at (1,1) {$b$};
\node (c) at (2,0) {$c$};
\draw[dashed, <->] (a) -- (b); 
\draw[dashed, ->] (c) -- (b); 

\end{tikzpicture}\\
& $ab; ac; bc$ & $cb; ab; ca$\\
\end{tabular}
\end{center}

\noindent We define two weakly connected graphs that are $\LNS$-unsuccessful, the {\em bush} and the {\em double bush}. The unsuccessful graph above will be a bush. In this section we prove that $\LNS$ is not weakly successful on bushes and double bushes. This will be is easy. in the next section we prove that $\LNS$ is weakly successful on any other weakly connected graph. This will be hard. First, we introduce additional tree terminology that we will use in these sections.

An initial gossip graph $G = (A,N,S)$ is a {\em tree} iff $(A, (N\setminus S)^{-1})$ is a directed rooted tree, i.e., from the perspective of the relation $N\setminus S$, there is unique node $r$ called {\em root} such that for any node $x$ there is a unique directed path from $x$ to $r$. Any node not accessible from other nodes is a {\em leaf}, and a path from a leaf to the root is a {\em branch}.

The {\em in-degree} of a node $a$ in a gossip graph is the number of incoming $N$-edges (excluding loops), i.e.~$\indeg(a) = | \{ b \neq a \mid Nba \} |$, and the {\em out-degree} of $a$ is the number of outgoing $N$-edges, i.e.~$\outdeg(a) = | \{ b \neq a \mid Nab \} |$. The {\em inout-degree} of a node $x$ is the sum of the in-degree and the out-degree: $\iodegr(x) = \indeg(x)+\outdeg(x)$.

A tree is a {\em bush} if the in-degree of the root is at least 2. The idea is, that a tree that is not a bush has a trunk. It has in-degree 1.

\begin{definition}[Bush, double bush] \label{reverse--tree} \label{double-reverse-tree}
A gossip graph $G = (\agents, N, S)$ is a \emph{bush} if{f} the graph $(\agents, N)$ is a bush. A {\em double bush} consists of two bushes $G_b + G_d$ that intersect in a leaf $c$ connected to both roots $b$ and $d$. Given that, $G$ is a \emph{double bush} if $\agents = \agents_{b} \cup \agents_{d}$ with $\agents_{b} \cap \agents_{d} = \{c\}$ such that the restrictions $G_b = G|\agents_{b}$ and $G_d = G|\agents_{d}$ are bushes and $\{(c,b),(c,d),(c,c)\} \subseteq N$. (See Figure \ref{fig.doublebushex1}.)
\end{definition}

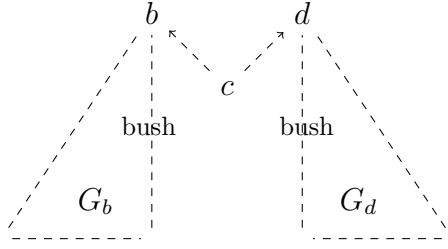
\begin{figure}
\begin{center}
\begin{tikzpicture}
\node (a1) at (-1,-2) {};
\node (a2) at (1,-2) {};
\node[rotate=0] (zz1) at (0.9,-.5) { \footnotesize{ bush }};
\node (a4) at (0.25,-1.5) {$G_{b}$};
\node (b) at (1,1) {$b$};
\node (c) at (2,0) {$c$};
\node (d) at (3,1) {$d$};
\node (e1) at (5,-2) {}; 
\node (e2) at (3,-2) {}; 
\node[rotate=0] (zz3) at (3,-.5) { \footnotesize{ bush }};
\node (e4) at (3.75,-1.5) {$G_{d}$};

\draw[dashed, -] (a1) -- (b); 
\draw[dashed, -] (a2) -- (b);
\draw[dashed, -] (a1) -- (a2); 
\draw[dashed, bend right, ->] (c) -- (b); 
\draw[dashed, bend left, ->] (c) -- (d); 
\draw[dashed, -] (e1) -- (d); 
\draw[dashed, -] (e2) -- (d); 
\draw[dashed, -] (e1) -- (e2); 
\end{tikzpicture}
\end{center}
\caption{A double bush consists of two bushes linked by a node $c$ connected to both roots.}
\label{fig.doublebushex1}
\end{figure}

\begin{lemma}\label{lemma-claims1-2}
Let $G = (\agents, N,S)$ be an initial gossip graph. If $G$ is a bush with root $r$, and $\sigma$ is an $\LNS$-sequence in $G$, then for any $x \in \agents$:
\begin{enumerate}
\item \label{itemp1} $G|{N^\sigma_x}$ is a tree.
\item \label{itemp2} $N^{\sigma}_{x}\setminus S^{\sigma}_{x} = \begin{cases}
\mbox{ root of } G|{N^\sigma_x} & \mbox{if } \neg S^{\sigma}xr\\
\emptyset & \mbox{otherwise }
\end{cases}$
\end{enumerate}
\end{lemma}

\begin{proof}
We prove both claims by induction on the length of $\sigma$. If $\sigma = \epsilon$ then the  claims hold by Definition \ref{reverse--tree}. For the inductive case of the proof, let us consider $\sigma;ab$ such that the inductive hypothesis holds for $\sigma$.
\begin{enumerate}
\item This is obviously true for $\sigma;ab$ for agents different from $a$ and $b$. Then, for agent $a$, we have $G|{N^{\sigma;ab}_a} = G|({N^\sigma_a \union N^\sigma_b})$ which is clearly cycle-free and connected. For agent $b$ this is as for agent $a$.

\item 
Again, it is obvious for agents other than $a,b$. Otherwise, since $\sigma;ab$ is an $\LNS$-sequence, $b$ is the root of $G|{N^{\sigma}_a}$. By induction hypothesis $\neg S^\sigma ar$. Then:
$$\begin{array}{cll}
& N^{\sigma;ab}_{a} \setminus S^{\sigma;ab}_{a} & \\
& & \\
= & (N^{\sigma}_{a} \cup N^{\sigma}_{b}) \setminus (S^{\sigma}_{a} \cup S^{\sigma}_{b}) & \mbox{(by Def.~\ref{def-call})}\\
& & \\
= & \begin{cases} (S^{\sigma}_{a} \cup \{b\} \cup S^{\sigma}_{b} \cup \{z\}) \setminus (S^{\sigma}_{a} \cup S^{\sigma}_{b}) & \mbox{if } N^{\sigma}_{b} \setminus S^{\sigma}_{b} = \{z\}\\
(S^{\sigma}_{a}
\cup \{b\} \cup S^{\sigma}_{b}) \setminus (S^{\sigma}_{a} \cup S^{\sigma}_{b}) & \mbox{if } N^{\sigma}_{b} \setminus S^{\sigma}_{b} = \varnothing\end{cases} & \mbox{(by Claim \ref{itemp2} for }\sigma)\\ & & \\
= & \begin{cases} \{z\} & \mbox{if } N^{\sigma}_{b} \setminus S^{\sigma}_{b} = \{z\}\\
\emptyset & \mbox{if } N^{\sigma}_{b} \setminus S^{\sigma}_{b} = \varnothing\end{cases} & \mbox{(by connectivity)}
\end{array}$$
\end{enumerate} \vspace{-.7cm}
\end{proof}

\begin{proposition} \label{Thm:reversebush}
Let an initial gossip graph $G = (\agents, N,S)$ be a bush. Then $\LNS$ is not weakly successful on $G$.
\end{proposition}

\begin{proof}
	Let $\sigma$ be $\LNS$-maximal on $G$, and let $yr$ and (subsequently) $ur$ be the first two calls in $\sigma$ to the root $r$ such that $y$ and $u$ are in subtrees generated by {\em different} branches to the root. We show that $\neg S^\sigma{yu}$. 

By Lemma~\ref{lemma-claims1-2}.\ref{itemp2}, after the call $yr$, $y$ will not make further calls. Obviously, $y$ cannot learn the secret of $u$ before call $ur$. Let $ay$ be any call to $y$ following $ur$. Let the sequence up to $ay$ be $\tau;ay$. Applying Lemma~\ref{lemma-claims1-2}.\ref{itemp2} we get $N^\tau_{a} \setminus S^\tau_{a} = \{y\}$. Therefore, $y$ is the root of $G|{N^\tau_a}$, and by connectivity of the tree, $\neg S^\tau{au}$. Therefore, also $\neg S^{\tau;ay}yu$.
\end{proof}

\begin{proposition}\label{Lem:2Jan-GeneralCase}
Let an initial gossip graph $G = (\agents, N,S)$ be a double bush. Then $\LNS$ is not weakly successful on $G$.
\end{proposition}

\begin{proof}
Let $b$ and $d$ be the two roots of $G$ and let $\sigma$ be a maximal $\LNS$-sequence. Without loss of generality, assume that the first call to $b$ takes place before the first call to $d$. Let $x$ be the first agent calling $b$. Clearly, $x$ has to be in $A_b$.

First consider $x \neq c$. After the call $xb$, agent $x$ will not make any further call, as we can apply Lemma \ref{lemma-claims1-2}.\ref{itemp1} to $G_b$. Agent $x$ will not learn the secret of $c$, because if $u$ calls $x$ after $xb$, then $u$ does not know the secret of $b$, so by connectivity $u$ must be in $G_{b}$, so that $N^*ux$. Thus, agent $u$ cannot inform $x$ of the secret of $c$.

If $x=c$, consider $G-cb+bc$. This is a bush. Let $\sigma'$ be the sequence obtained by replacing in $\sigma$ call $cb$ with $bc$. Then $(G-cb+bc)^{\sigma'}=G^{\sigma}$ (and $\sigma'$ is also maximal). By Proposition \ref{Thm:reversebush}, $\LNS$ is unsuccessful on $G-cb+bc$, and so is also unsuccessful on $G$.
\end{proof}

\subsection{Where $\LNS$ is weakly successful} \label{subs.3}

In this section we prove that a gossip graph that is neither a bush nor a double bush is weakly successful for $\LNS$. The setup of this lengthy proof is as follows.

In gossip graphs that are trees, (almost) every $\LNS$-permitted call that is made, makes a new call $\LNS$-permitted. We use this to define a particular $\LNS$-sequence generating procedure, called {\em bottom-up call sequence} (Definition \ref{lexicographic-calls}), that will then be used profusely in the subsequent technical results of this section. We then show that bushes and double bushes are maximal for the property of being $\LNS$-unsuccessful: adding a new edge or a new node destroys this property (Lemmas~\ref{Lem:1Jan-AddEdge}--\ref{Lem:2JanPlusNode}). A further illustration of that is in Figure \ref{fig.doublebushex2}. With these intermediary results we can then finally prove that a gossip graph that is neither a bush nor a double bush is weakly successful for $\LNS$ (Proposition~\ref{theo.takki}). The proof is by induction on the number of nodes and edges of the graph. This proof consists of many cases, of which a crucial case is supported by an additional Lemma \ref{lem:Case-2.2.3.2}.

\begin{figure}
\begin{center}
\begin{tabular}{ccc}
\begin{tikzpicture}
\node (a) at (0,0) {$a$};
\node (b) at (1,1) {$b$};
\node (c) at (2,0) {$c$};
\node (d) at (3,1) {$d$};
\node (e) at (4,0) {$e$};

\draw[dashed, ->] (a) -- (b); 
\draw[dashed, ->] (c) -- (b); 
\draw[dashed, ->] (c) -- (d); 
\draw[dashed, ->] (e) -- (d); 
\end{tikzpicture}
&
\begin{tikzpicture}
\node (a) at (0,0) {$a$};
\node (b) at (1,1) {$b$};
\node (c) at (2,0) {$c$};
\node (d) at (3,1) {$d$};
\node (e) at (4,0) {$e$};
\draw[dashed, <->] (a) -- (b); 
\draw[dashed, ->] (c) -- (b); 
\draw[dashed, ->] (c) -- (d); 
\draw[dashed, ->] (e) -- (d); 
\end{tikzpicture}&
\begin{tikzpicture}
\node (a) at (0,0) {$a$};
\node (b) at (1,1) {$b$};
\node (c) at (2,0) {$c$};
\node (d) at (3,1) {$d$};
\node (e) at (4,0) {$e$};
\draw[dashed, ->] (a) -- (b); 
\draw[dashed, <->] (c) -- (b); 
\draw[dashed, ->] (c) -- (d); 
\draw[dashed, ->] (e) -- (d); 
\end{tikzpicture}
\\
unsuccessful & $cb;ab;cd;ed;ad;bd;ca;ea$ & $ab;cd;ed;db;cb;ac;eb$ \\ \ \\
\begin{tikzpicture}
\node (a) at (0,0) {$a$};
\node (b) at (1,1) {$b$};
\node (c) at (2,0) {$c$};
\node (d) at (3,1) {$d$};
\node (e) at (4,0) {$e$};
\draw[dashed, ->] (a) -- (b); 
\draw[dashed, ->] (a) -- (c); 
\draw[dashed, ->] (c) -- (b); 
\draw[dashed, ->] (c) -- (d); 
\draw[dashed, ->] (e) -- (d); 
\end{tikzpicture}
&
\begin{tikzpicture}
\node (a) at (0,0) {$a$};
\node (b) at (1,1) {$b$};
\node (c) at (2,0) {$c$};
\node (d) at (3,1) {$d$};
\node (e) at (4,0) {$e$};
\draw[dashed, ->] (a) -- (b); 
\draw[dashed, ->] (c) -- (a); 
\draw[dashed, ->] (c) -- (b); 
\draw[dashed, ->] (c) -- (d); 
\draw[dashed, ->] (e) -- (d); 
\end{tikzpicture}&
\begin{tikzpicture}
\node (a) at (0,0) {$a$};
\node (b) at (1,1) {$b$};
\node (c) at (2,0) {$c$};
\node (d) at (3,1) {$d$};
\node (e) at (4,0) {$e$};
\draw[dashed, ->] (a) -- (b); 
\draw[dashed, ->] (a) -- (d); 
\draw[dashed, ->] (c) -- (b); 
\draw[dashed, ->] (c) -- (d); 
\draw[dashed, ->] (e) -- (d); 
\end{tikzpicture}\\
$ab;cd;ed;db;cb;ac;eb$ & $ab;cd;ed;da;ca;eb$ & $ab;cd;ed;ad;bd;cb;eb$ \\ \ \\
\begin{tikzpicture}
\node (a) at (0,0) {$a$};
\node (b) at (1,1) {$b$};
\node (c) at (2,0) {$c$};
\node (d) at (3,1) {$d$};
\node (e) at (4,0) {$e$};
\draw[dashed, ->] (a) -- (b); 
\draw[dashed, ->] (d) -- (a); 
\draw[dashed, ->] (c) -- (b); 
\draw[dashed, ->] (c) -- (d); 
\draw[dashed, ->] (e) -- (d); 
\end{tikzpicture} &
\begin{tikzpicture}
\node (a) at (0,0) {$a$};
\node (b) at (1,1) {$b$};
\node (c) at (2,0) {$c$};
\node (d) at (3,1) {$d$};
\node (e) at (4,0) {$e$};
\draw[dashed, ->] (a) -- (b); 
\draw[dashed, ->] (b) -- (d); 
\draw[dashed, ->] (c) -- (b); 
\draw[dashed, ->] (c) -- (d); 
\draw[dashed, ->] (e) -- (d); 
\end{tikzpicture}& 
\begin{tikzpicture}
\node (a) at (0,0) {$a$};
\node (b) at (1,1) {$b$};
\node (c) at (2,0) {$c$};
\node (d) at (3,1) {$d$};
\node (e) at (4,0) {$e$};
\draw[dashed, ->] (a) -- (b); 
\draw[dashed, ->] (c) -- (b); 
\draw[dashed, ->] (c) -- (d); 
\draw[dashed, ->] (e) -- (d); 
\path (a) edge [dashed, bend right = 10 ,->] (e);
\end{tikzpicture}\\
$ab;cd;ed;da;ca;eb$ & $ab;cd;ed;ad;bd;cb;eb$ & $ab;ae;be;cb;bd;ad;cd;ed$\\
\end{tabular}
\end{center}
\caption{On the top-left a double bush of minimum size. The other figures demonstrate that adding an edge to this graph makes it weakly successful. A similar exercise demonstrates that adding a node $f$ and an edge also makes the graph weakly successful, except when the edge is $(f,a)$ or $(f,e)$, because it then still is a double bush.}
\label{fig.doublebushex2}
\end{figure}
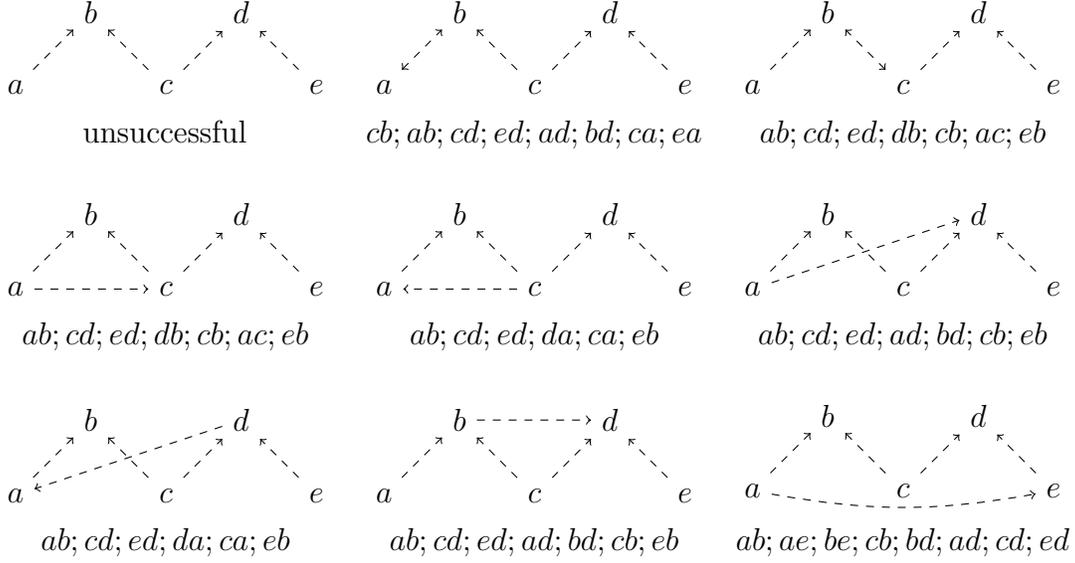

\begin{definition}[Bottom-up call sequence]\label{lexicographic-calls}
Let $G=(\agents,N,S)$ be a (possibly non-initial) gossip graph such that $N\setminus S$ is a tree, with: root $r$ \weg{depth $n$} and set of leaves $B$. We define a \emph{bottom-up call sequence} $\sigma(k)$ and the {\em frontline} $B(k)$ by simultaneous induction on $k$. We should see the set $B'(k)$ below as the set of {\em successor} nodes of $B(k)$.
\[ \begin{array}{lll}
\sigma(0) &=& \epsilon \\
\sigma(k+1) &=& \sigma(k);\tau(k) \\
B(0) &=& B \\
B(k+1) &=& B(k) \union B'(k) \\
\ \vspace{-.2cm} \\ \text{where} \\ \ \vspace{-.2cm} \\
B'(k) & = & \Union_{b \in B(k)} N^{\sigma(k)}_b \setminus S^{\sigma(k)}_b\\
\tau(k) &=& \text{a maximal $\LNS$-sequence between members of $B(k)$ and members of $B'(k)$}
\end{array}\]
\end{definition}
If $n$ is the depth of the tree, then $\sigma(k) = \sigma(n)$ for $k>n$. It will also be obvious that:
\begin{proposition}
A bottom-up call sequence is $\LNS$-permitted. If $n$ is the depth of the tree, then $\sigma(n)$ is $\LNS$-maximal.
\end{proposition}
The graph below illustrates the execution of the bottom-up call sequence $\tau(1);\tau(2);\tau(3)$.
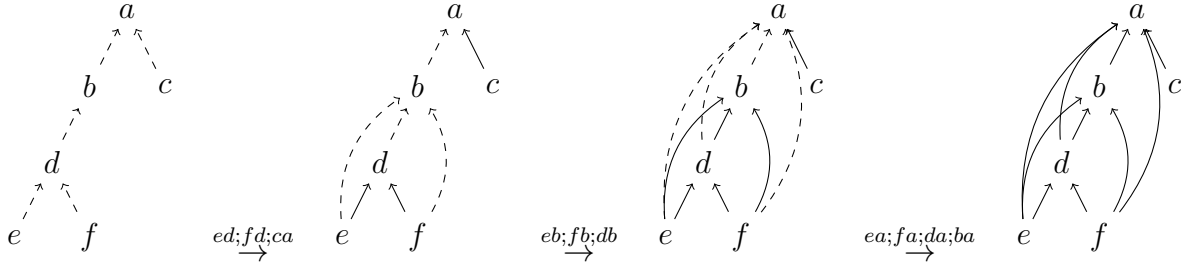
\begin{center}
\begin{tikzpicture}
\node (a) at (0,4) {$a$};
\node (b) at (-.5,3) {$b$};
\node (c) at (.5,3) {$c$};
\node (d) at (-1,2) {$d$};
\node (e) at (-1.5,1) {$e$};
\node (f) at (-.5,1) {$f$};
\path[dashed, ->]
(b) edge (a)
(c) edge (a)
(d) edge (b)
(e) edge (d)
(f) edge (d);
\end{tikzpicture}
~~$\stackrel {ed;fd;ca} \imp$~~
\begin{tikzpicture}
\node (a) at (0,4) {$a$};
\node (b) at (-.5,3) {$b$};
\node (c) at (.5,3) {$c$};
\node (d) at (-1,2) {$d$};
\node (e) at (-1.5,1) {$e$};
\node (f) at (-.5,1) {$f$};
\path[dashed, ->]
(b) edge (a)
(d) edge (b)
(e) edge [bend left] (b)
(f) edge [bend right] (b);
\path[->]
(c) edge (a)
(e) edge (d)
(f) edge (d);
\end{tikzpicture}
~~$\stackrel {eb;fb;db} \imp$~~
\begin{tikzpicture}
\node (a) at (0,4) {$a$};
\node (b) at (-.5,3) {$b$};
\node (c) at (.5,3) {$c$};
\node (d) at (-1,2) {$d$};
\node (e) at (-1.5,1) {$e$};
\node (f) at (-.5,1) {$f$};
\path[dashed, ->]
(b) edge (a)
(d) edge [bend left] (a)
(e) edge [bend left] (a)
(f) edge [bend right] (a);
\path[->]
(d) edge (b)
(c) edge (a)
(e) edge (d)
(f) edge (d)
(e) edge [bend left] (b)
(f) edge [bend right] (b);
\end{tikzpicture}
~~$\stackrel {ea;fa;da;ba} \imp$~~
\begin{tikzpicture}
\node (a) at (0,4) {$a$};
\node (b) at (-.5,3) {$b$};
\node (c) at (.5,3) {$c$};
\node (d) at (-1,2) {$d$};
\node (e) at (-1.5,1) {$e$};
\node (f) at (-.5,1) {$f$};
\path[->]
(b) edge (a)
(d) edge [bend left] (a)
(e) edge [bend left] (a)
(f) edge [bend right] (a)
(d) edge (b)
(c) edge (a)
(e) edge (d)
(f) edge (d)
(e) edge [bend left] (b)
(f) edge [bend right] (b);
\end{tikzpicture}
\end{center}

\begin{lemma}\label{fact-reverse-tree}
Let an initial gossip graph $G$ be a tree with root $r$.
\begin{enumerate}
\item \label{oneoneone} There is a $\LNS$-sequence $\sigma$ after which the root $r$ is an expert and every agent knows the secret of $r$.
\item \label{twotwotwo} If $r$ has exactly one predecessor, then there is a $\LNS$-sequence that is successful on $G$.
\end{enumerate}
\end{lemma}
\begin{proof}
Let $n$ be the maximum depth of $G$. 
\begin{enumerate}
\item In a maximal bottom-up call sequence all other agents have called the root $r$.
\item Let $r'$ be the predecessor of $r$, and let $G' = G - r - r'r$. Then $G'$ is a tree with root $r'$. Let $\sigma$ be a maximal bottom-up call sequence in $G'$, and let $\sigma'$ be the sequence consisting of calls from all agents except $r$ and $r'$, to $r$. Then $\sigma; r'r; \sigma'$ is successful on $G$. 
\end{enumerate} \vspace{-.7cm}
\end{proof}
 
\begin{lemma}\label{Lem:1Jan-AddEdge}
Let $G = (A,N,S)$ be a bush and let $x,y \in A$ such that $\neg Nxy$. Then there is an $\LNS$-sequence that is successful on $G+xy$.
\end{lemma}

\begin{proof}
If $\neg N^{\ast}xy$, then $y$ is not the root $r$ of $G$. Let $t$ be the (unique) $N$-successor of $y$. Graph $G - yt$ consists of two disjoint gossip graphs $H_{1}$ and $H_{2}$ with domains $A_1$ and $A_2$, respectively, such that $x, r, t \in A_{1}$ and $y \in A_{2}$. The generated tree with root $x$ is a sub-graph of $H_1$; let $\sigma^x$ be a maximal bottom-up call sequence in that tree. Let $\sigma^{xr}$ be the sequence where $x$ calls all the agents in the path from $x$ to $r$ in $H_1$. Observe that in $H_{1}^{\sigma^{x};\sigma^{xr}}$, $N_{1}^{\sigma^{x};\sigma^{xr}}\setminus S_{1}^{\sigma^{x};\sigma^{xr}}$ is a tree (Lemma \ref{lemma-claims1-2}.\ref{itemp1}). Let $\sigma^{1}$ and $\sigma^{2}$ be maximal bottom-up call sequences for $H_{1}^{\sigma^x;\sigma^{xr}}$ and $H_{2}$, respectively. The call sequence \[ \sigma^{x};\sigma^{xr}; \sigma^{1}; \sigma^{2}; ry; (A_{1}-r)y; (A_{2}-y)t \] is $\LNS$-permitted and successful on $G + xy$; where, as usual, $Xz$ stands for a sequence of calls from each agent in $X$ to agent $z$. After $\sigma^{x};\sigma^{xr}; \sigma^{1}$, the root $r$ is an expert in $H_{1}$ and all the agents in $H_{1}$ have the number of $y$. After $\sigma^{2};ry$, both $r$ and $y$ are experts and all agents in $H_{2}$ have the number of $t$. After $(A_{1}-r)y$ all agents in $A_{1}$ are experts, in particular $t$. At the end, all agents are experts.

If $N^{\ast}xy$, then, since $x \neq y$, $x$ is not the root. Let $t$ be the $N$-successor of $x$. Write $G + xy$ as $(G - xt + xy) + xt$. Graph $(G - xt + xy)$ is a bush, and there is no path from $x$ to $t$ in that graph. We therefore can apply the part `if $\neg N^{\ast}xy$' of this proof, on $(G - xt + xy)$ and $xt$, instead of on $G$ and $xy$.
\begin{center}
\begin{tabular}{c@{\hspace{14mm}}c}
\begin{tikzpicture}
\draw[-,thick] (0.25, -2.5) -- (2.75,-2.5) -- (0.90,0.2) -- (-.34,-1.5) -- (0.25, -2.5);
\draw[-,thick] (-.5,-1.68) -- (-1,-2.5) -- (0, -2.5) -- (-.5,-1.68);
\node (b) at (1,.5) {$r$};
\node (x) at (0.41,0) {$x$};
\node (y) at (-.76,-1.63) {$y$};
\node (t) at (.32,-1) {$t$};
\draw[fill=black] (-.5,-1.67) circle (.07cm and .07cm); 
\draw[fill=black] (0.05,-0.95) circle (.07cm and .07cm); 
\draw[fill=black] (.57,-0.25) circle (.07cm and .07cm); 
\draw[fill=black] (.9,0.2) circle (.07cm and .07cm); 
\draw[dashed,thick, ->] (-.5,-1.47) -- (-0.15,-1.0); 
\draw[dashed,thick, ->] (.47,-0.25) to [bend right] (-.5,-1.33); 
\node at (-.50,-2.2) {\footnotesize $H_2$};
\node at (.9,-.6) {\footnotesize $H_1$};
\end{tikzpicture} & 
\begin{tikzpicture}
\draw[-,thick] (-1,-2.5) -- (2.75,-2.5) -- (0.90,0.2) -- (-1,-2.5); 
\node (b) at (1,.5) {$r$};
\node (y) at (.1,-0.2) {$y$};
\node (x) at (-1.1,-2.2) {$x$};
\node (z) at (.1,-1.2) {$t$};
\draw[fill=black] (.6,-0.2) circle (.07cm and .07cm); 
\draw[fill=black] (-.1,-1.2) circle (.07cm and .07cm); 
\draw[fill=black] (-.8,-2.2) circle (.07cm and .07cm); 
\draw[fill=black] (.9,0.2) circle (.07cm and .07cm); 
\draw[dashed,thick, ->] (-.8,-2.01) -- (-.25,-1.2); 
\draw[dashed,thick, ->] (-.8,-2.01) to [bend left] (.4,-0.3); 
\end{tikzpicture}\\[1mm]
Case $\neg N^{\ast}xy$ & Case $N^{\ast}xy$ 
\end{tabular}
\end{center} \vspace{-.7cm}
\end{proof}

\begin{lemma}\label{lemma-reverse-tree-1edge}
Let $G= (\agents, N, S)$ be a bush, $y \notin \agents$, $x \in \agents$, and let $e$ be an edge between $y$ and $x$. If $G+y+e$ is not a bush, there is an $\LNS$-sequence that is successful on $G$. 
\end{lemma} 

\begin{proof}
If $e= yx$ then $G+y+yx$ is again a bush, in which case the lemma is trivially true. So let $e = xy$. If $x = r$ we can use Lemma \ref{fact-reverse-tree}.\ref{twotwotwo} and we are done. Let $x \neq r$. Consider $H=G+y+yx$. Then, $H$ is a bush. Hence, we can apply Lemma~\ref{Lem:1Jan-AddEdge} to obtain a successful $\LNS$-sequence on $H+xy$, say $\sigma$. If $\sigma$ does not contain the call $yx$, then $\sigma$ is also $\LNS$-permitted in $G+y+xy$, and hence also successful. Otherwise, replace $yx$ in $\sigma$ with $xy$. Clearly, the resulting sequence $\sigma'$ is again $\LNS$-permitted in $G+y+xy$ and also successful.
\end{proof}

\begin{lemma}\label{Lem:2JanPlusEdge}
Let $G$ be a double bush, and let $\neg Nxy$. Then $G+xy$ has a successful $\LNS$-sequence.
\end{lemma}
\begin{proof}
Let $G = (A,N,S)$. Recall the different components of a double bush with roots $b,d$ and a common node $c$: $\agents = \agents_{b} \cup \agents_{d}$ and $\agents_{b} \cap \agents_{d} = \{c\}$. As these components are alike, it suffices to consider 3 cases (see Figure \ref{CasesJanPlusEdge}).

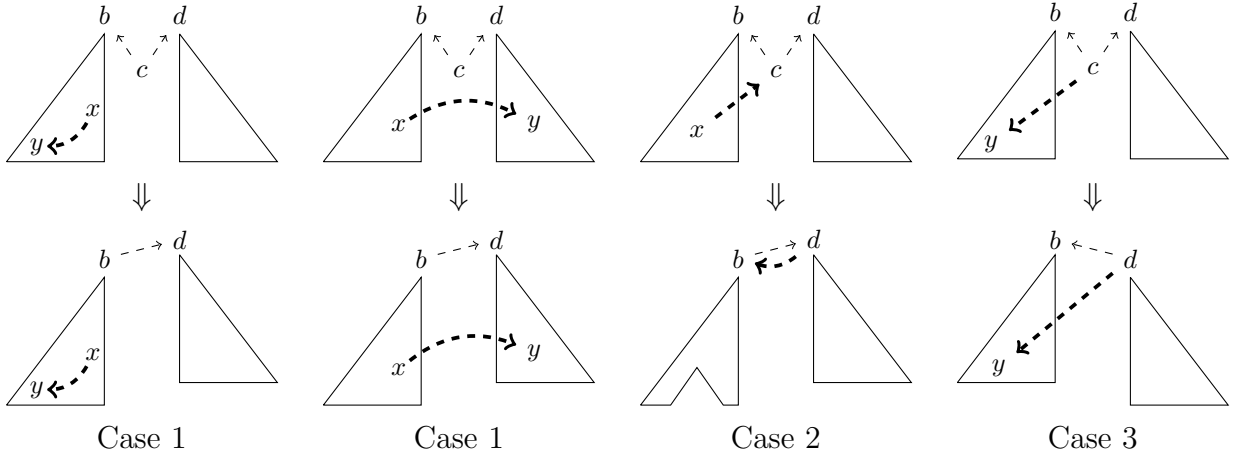
\begin{figure}
\begin{tabular}{c@{\hspace{6mm}}c@{\hspace{6mm}}c@{\hspace{6mm}}c}
\begin{tikzpicture}
\node (b) at (-.5,.75) {\footnotesize $b$};
\node (c) at (0,0) {\footnotesize $c$};
\node (d) at (.5,.75) {\footnotesize $d$};
\node[inner sep=.4mm] (y2) at (-1.4,-1) {\footnotesize $y$};
\node[inner sep=.8mm] (x) at (-.65,-.52) {\footnotesize $x$}; 
\draw[-] (-1.8,-1.2) -- (-.5,-1.2) -- (-.5,.5) -- (-1.8,-1.2);
\draw[-] (.5,.5) -- (.5,-1.2) -- (1.8,-1.2) -- (.5,.5);
\draw[dashed, ->] (c) -- (b); 
\draw[dashed, ->] (c) -- (d);
\path [line width=.05cm, dashed, bend left]
 (x) edge [->] node { } (y2);
\end{tikzpicture} &
\begin{tikzpicture}
\node (b) at (-.5,.75) {\footnotesize $b$};
\node (c) at (0,0) {\footnotesize $c$};
\node (d) at (.5,.75) {\footnotesize $d$};
\node (y) at (1,-.7) {\footnotesize $y$};
\node[inner sep=.4mm] (x) at (-.8,-.72) {\footnotesize $x$};
\draw[-] (-1.8,-1.2) -- (-.5,-1.2) -- (-.5,.5) -- (-1.8,-1.2);
\draw[-] (.5,.5) -- (.5,-1.2) -- (1.8,-1.2) -- (.5,.5);
\draw[dashed, ->] (c) -- (b); 
\draw[dashed, ->] (c) -- (d);
\path [line width=.05cm, dashed, bend left]
 (x) edge [->] node { } (y);
\end{tikzpicture} &
\begin{tikzpicture}
\node (b) at (-.5,.75) {\footnotesize $b$};
\node (c) at (0,0) {\footnotesize $c$};
\node (d) at (.5,.75) {\footnotesize $d$};
\node (x) at (-1.05,-.8) {\footnotesize $x$};
\draw[-] (-1.8,-1.2) -- (-.5,-1.2) -- (-.5,.5) -- (-1.8,-1.2);
\draw[-] (.5,.5) -- (.5,-1.2) -- (1.8,-1.2) -- (.5,.5);
\draw[dashed, ->] (c) -- (b); 
\draw[dashed, ->] (c) -- (d); 
\path [line width=.05cm, dashed]
 (x) edge [->] node { } (c);
\end{tikzpicture} &
\begin{tikzpicture}
\node (b) at (-.5,.75) {\footnotesize $b$};
\node (c) at (0,0) {\footnotesize $c$};
\node (d) at (.5,.75) {\footnotesize $d$};
\node (y) at (-1.35,-1) {\footnotesize $y$};
\draw[-] (-1.8,-1.2) -- (-.5,-1.2) -- (-.5,.5) -- (-1.8,-1.2);
\draw[-] (.5,.5) -- (.5,-1.2) -- (1.8,-1.2) -- (.5,.5);
\draw[dashed, ->] (c) -- (b); 
\draw[dashed, ->] (c) -- (d); 
\path [line width=.05cm, dashed]
 (c) edge [->] node { } (y);
\end{tikzpicture}\\[.5mm] 
$\Downarrow$ & $\Downarrow$ & $\Downarrow$ & $\Downarrow$\\[.5mm]

\begin{tikzpicture}
\node (b) at (-.5,.75) {\footnotesize $b$};
\node (d) at (.5,1) {\footnotesize $d$};
\node[inner sep=.4mm] (y2) at (-1.4,-1) {\footnotesize $y$};
\node[inner sep=.8mm] (x) at (-.65,-.52) {\footnotesize $x$}; 
\draw[-] (-1.8,-1.2) -- (-.5,-1.2) -- (-.5,.5) -- (-1.8,-1.2);
\draw[-] (.5,.8) -- (.5,-.9) -- (1.8,-.9) -- (.5,.8);
\draw[dashed, ->] (b) -- (d);
\path [line width=.05cm, dashed, bend left]
 (x) edge [->] node { } (y2);
\end{tikzpicture} &
\begin{tikzpicture}
\node (b) at (-.5,.75) {\footnotesize $b$};
\node (d) at (.5,1) {\footnotesize $d$};
\node (y) at (1,-.5) {\footnotesize $y$};
\node[inner sep=.4mm] (x) at (-.8,-.72) {\footnotesize $x$};
\draw[-] (-1.8,-1.2) -- (-.5,-1.2) -- (-.5,.5) -- (-1.8,-1.2);
\draw[-] (.5,.8) -- (.5,-.9) -- (1.8,-.9) -- (.5,.8);
\draw[dashed, ->] (b) -- (d);
\path [line width=.05cm, dashed, bend left]
 (x) edge [->] node { } (y);
\end{tikzpicture} &
\begin{tikzpicture}
\node (b) at (-.5,.75) {\footnotesize $b$};
\node (d) at (.5,1.01) {\footnotesize $d$};
\draw[-,white, fill=white] (-.1, -.6) -- (.1,-.6) -- (.1,-1) -- (-.1,-1) -- (-.1, -.6);
\draw[-] (.5,.8) -- (.5,-.9) -- (1.8,-.9) -- (.5,.8);
\draw[-] (-1.8,-1.2) -- (-1.4, -1.2) -- (-1.05, -.7) -- (-.7, -1.2) -- (-.5,-1.2) -- (-.5,.5) -- (-1.8,-1.2);
\draw[dashed, ->] (b) -- (d);
\path [line width=.05cm, dashed, bend left]
(d) edge [->] node { } (b);
\end{tikzpicture}
&
\begin{tikzpicture}
\node (b) at (-.5,1) {\footnotesize $b$};
\node (d) at (.5,.75) {\footnotesize $d$};
\node (y) at (-1.25,-.7) {\footnotesize $y$};
\draw[-] (-1.8,-.9) -- (-.5,-.9) -- (-.5,.8) -- (-1.8,-.9);
\draw[-] (.5,.5) -- (.5,-1.2) -- (1.8,-1.2) -- (.5,.5);
\draw[dashed, ->] (d) -- (b); 
\path [line width=.05cm, dashed]
 (d) edge [->] node { } (y);
\end{tikzpicture}\\[.5mm]
Case 1 & Case 1 & Case 2 & Case 3
\end{tabular}
\caption{The three cases used in the proof of Lemma \ref{Lem:2JanPlusEdge}, and their modifications.}
\label{CasesJanPlusEdge}
\end{figure}

(Case 1: $x \in A^{b} -c$ and $y \neq c$.) In this case $y$ may be in part $G_b$ or in part $G_d$ of the graph, but that does not matter for the proof. 

If $xy = bd$, then $G - cb + bd$ is a bush, so by Lemma~\ref{Lem:1Jan-AddEdge}, we obtain a successful $\LNS$-sequence $\sigma$ on $(G - cb +bd) +cb$, i.e., on $G+bd$. 

If $xy \neq bd$, consider $H = G - c - cb - cd + bd$. As $H$ is a bush, we can apply Lemma~\ref{lemma-reverse-tree-1edge} to obtain an $\LNS$-sequence $\sigma$ successful on $H + xy$. The sequence $cb; \sigma; cd$ is $\LNS$-permitted and successful on $G + xy = (H + xy) + c + cb + cd - bd$.

(Case 2: $x \in A^{b}-c$ and $y = c$.) If $x = b$, $G - cb + bc$ is a bush, so we can apply Lemma~\ref{Lem:1Jan-AddEdge} to obtain an successful $\LNS$-sequence $\sigma$ on $(G - cb + bc) + cb = G + bc$. If $x \neq b$, let $G^{x}$, for agents $A^x$, be the subtree of $G$ with root $x$, and let $\sigma^{x}$ be a maximal bottom-up call sequence on $G^{x}$. The graph $H = G - G^{x} - c - cb - cd + bd$ is a bush, so we can apply Lemma~\ref{Lem:1Jan-AddEdge} to obtain a successful $\LNS$-sequence $\sigma$ on $H+db$. On $G + xc$, sequence $\sigma^{x}; xc; cb; xd; \sigma; cd; xb; (A^x-x)c$ is $\LNS$-successful, where, as before, $(A^x-x)c$ denotes a sequence from all agents in $(A^x-x)$ to agent $c$. From Lemma \ref{fact-reverse-tree} it follows that after $\sigma^{x}$ everybody knows the secret of $x$ and thus also the number of $c$, so the calls $(A^x-x)c$ can indeed take place.

(Case 3: $x = c$ and $y \in A^{b}$.) Let $H = G - c - cb - cd + db$. This is a bush. We apply Lemma~\ref{Lem:1Jan-AddEdge} to obtain a successful $\LNS$-sequence $\sigma$ on $H + dy$. On $G + cy$ the sequence $cd; \sigma; cb$ is $\LNS$-successful.
\end{proof}

 \begin{lemma}\label{Lem:2JanPlusNode}
Let $G = (A,N,S)$ be a double bush, $y \notin A$, $x \in A$, and $e$ an edge between $x$ and $y$. There is a successful $\LNS$ call sequence on $G+y+e$ unless it is a double bush.
\end{lemma}
\begin{proof}
Let $\agents=\agents_{b} \cup \agents_{d}$ and $\agents_{b} \cap \agents_{d} = c$, where $b,d$ are the two roots and $c$ is the common leaf. Without loss of generality, suppose $x \notin \agents_{d}$. Consider the bush $H = G - c - cb - cd + bd$. There are 3 cases to consider in this proof: $e = xy$, $e = yc$, and $e = cy$.
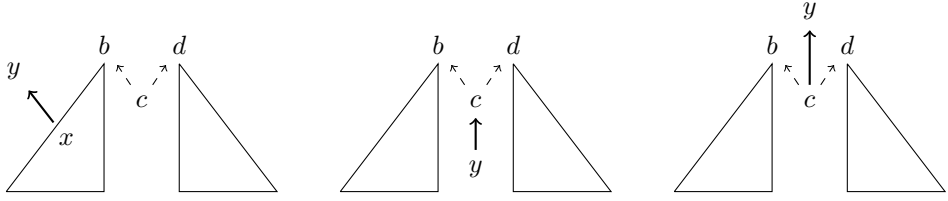
\begin{center}
\begin{tabular}{c@{\qquad}c@{\qquad}c}
\begin{tikzpicture}
\node (b) at (-.5,.75) {\footnotesize $b$};
\node (c) at (0,0) {\footnotesize $c$};
\node (d) at (.5,.75) {\footnotesize $d$};
\node (x) at (-1,-.5) {\footnotesize $x$};
\node (y) at (-1.7,.4) {\footnotesize $y$};
\draw[-] (-1.8,-1.2) -- (-.5,-1.2) -- (-.5,.5) -- (-1.8,-1.2);
\draw[-] (.5,.5) -- (.5,-1.2) -- (1.8,-1.2) -- (.5,.5);
\draw[dashed, ->] (c) -- (b); 
\draw[dashed, ->] (c) -- (d); 
\path [thick]
 (x) edge [->] node { } (y);
\end{tikzpicture} &
\begin{tikzpicture}
\node (b) at (-.5,.75) {\footnotesize $b$};
\node (c) at (0,0) {\footnotesize $c$};
\node (d) at (.5,.75) {\footnotesize $d$};
\node (y) at (0,-.9) {\footnotesize $y$};
\draw[-] (-1.8,-1.2) -- (-.5,-1.2) -- (-.5,.5) -- (-1.8,-1.2);
\draw[-] (.5,.5) -- (.5,-1.2) -- (1.8,-1.2) -- (.5,.5);
\draw[dashed, ->] (c) -- (b); 
\draw[dashed, ->] (c) -- (d); 
\path [thick]
 (y) edge [->] node { } (c);
\end{tikzpicture} &
\begin{tikzpicture}
\node (b) at (-.5,.75) {\footnotesize $b$};
\node (c) at (0,0) {\footnotesize $c$};
\node (d) at (.5,.75) {\footnotesize $d$};
\node (y) at (0,1.2) {\footnotesize $y$};
\draw[-] (-1.8,-1.2) -- (-.5,-1.2) -- (-.5,.5) -- (-1.8,-1.2);
\draw[-] (.5,.5) -- (.5,-1.2) -- (1.8,-1.2) -- (.5,.5);
\draw[dashed, ->] (c) -- (b); 
\draw[dashed, ->] (c) -- (d); 
\path [thick]
 (c) edge [->] node { } (y);
\end{tikzpicture}
\end{tabular}
\end{center} 

If $x \in \agents^{b}-c$ and $e=xy$, apply Lemma~\ref{lemma-reverse-tree-1edge} on $H$ to obtain a successful $\LNS$-sequence $\sigma$ on $H + y + xy$. The sequence $cb; \sigma; cd$ is $\LNS$-permitted and successful on $G + y + xy$.

If $x=c$ and $e=yx$, apply Lemma \ref{Lem:1Jan-AddEdge} to obtain a successful $\LNS$-sequence $\sigma$ on $H + db$. The sequence $yc; cb; yd; \sigma; cd; yb$ is $\LNS$-permitted and successful on $G + y + yc$. 

If $x=c$ and $e=xy$, apply Lemma~\ref{lemma-reverse-tree-1edge} to obtain a successful $\LNS$-sequence $\sigma$ on $H + y + by$. The sequence $cb; \sigma; cd$ is $\LNS$-permitted and successful on $G + y + cy$.
\end{proof}
 
\begin{proposition} \label{theo.takki}
Let $G$ be a weakly connected graph. If $G$ is neither a bush nor a double bush, then $G$ is weakly successful on $\LNS$.
\end{proposition}
\begin{proof}
The proof is by strong induction on the sum of the number $n$ of nodes and the number $m$ of edges (not counting loops) of such gossip graphs $G$. If $n+m = 1$, $G$ is a single agent gossip graph, on which the empty sequence is successful. Let now $G$ be a gossip graph with $n+m = k+1$ nodes and edges, and assume the proposition to be proved for $l \leq k$ nodes and edges (call such an inductive case $\phi(l)$).

If $G$ has an edge $e$ which is not a bridge, then $G-e$ is weakly connected. In case this $G-e$ is neither a bush nor a double bush, then by inductive hypothesis (for $k$) there is a successful $\LNS$-sequence $\sigma$ on $G-e$. This $\sigma$ is also successful on $G$. In case $G-e$ is a bush we apply Lemma \ref{Lem:1Jan-AddEdge}, and in case $G-e$ is a double bush we apply Lemma \ref{Lem:2JanPlusEdge}, in order to obtain a successful $\LNS$-sequence $\sigma$ on $(G-e)+e=G$.

Let now every edge in $G$ be a bridge. If $G$ does not have a node with out-degree greater than 1, then $G$ is a tree. The input-degree of the root of this tree must be 1, because we assumed that $G$ is not a bush. By Lemma \ref{fact-reverse-tree}.\ref{twotwotwo} there must be a successful $\LNS$-sequence on $G$. However, if $G$ has a node $s$ with out-degree greater than 1, then, since the undirected graph underlying $G$ is a tree (because every edge is a bridge), there are at least two nodes with inout-degree 1. Let $y$ be such a node with inout-degree 1. We need to consider three cases (pictured from left to right, itemized from top to bottom):

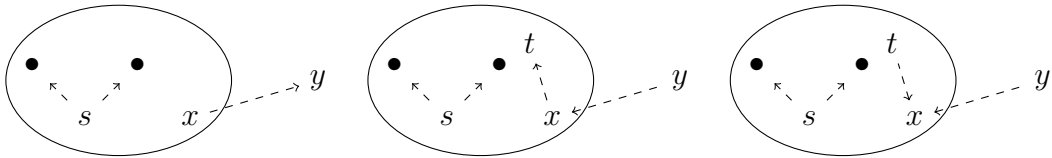
\begin{center}
\begin{tabular}{c@{\quad}c@{\quad}c}
\begin{tikzpicture}
\draw (0.45,.5) circle (1.5cm and 1cm);
\node (c) at (0,0) {$s$};
\node (d) at (.7,.7) {\textbullet};
\node (b) at (-.7,.7) {\textbullet};
\node (x) at (1.40,0) {$x$};
\node (y) at (3.1,.5) {$y$};
\draw[dashed, ->] (x) -- (y); 
\draw[dashed, ->] (c) -- (d); 
\draw[dashed, ->] (c) -- (b); 
\end{tikzpicture} & 
\begin{tikzpicture}
\draw (0.45,.5) circle (1.5cm and 1cm);
\node (c) at (0,0) {$s$};
\node (d) at (.7,.7) {\textbullet};
\node (b) at (-.7,.7) {\textbullet};
\node (x) at (1.40,0) {$x$};
\node (t) at (1.1,1) {$t$};
\node (y) at (3.1,.5) {$y$};
\draw[dashed, <-] (x) -- (y); 
\draw[dashed, ->] (x) -- (t); 
\draw[dashed, ->] (c) -- (d); 
\draw[dashed, ->] (c) -- (b); 
\end{tikzpicture}
&
\begin{tikzpicture}
\draw (0.45,.5) circle (1.5cm and 1cm);
\node (c) at (0,0) {$s$};
\node (d) at (.7,.7) {\textbullet};
\node (b) at (-.7,.7) {\textbullet};
\node (x) at (1.40,0) {$x$};
\node (t) at (1.1,1) {$t$};
\node (y) at (3.1,.5) {$y$};
\draw[dashed, <-] (x) -- (y); 
\draw[dashed, <-] (x) -- (t); 
\draw[dashed, ->] (c) -- (d); 
\draw[dashed, ->] (c) -- (b); 
\end{tikzpicture}
\end{tabular}
\end{center}

\begin{itemize}
\item If there is an edge $xy$ with $\outdeg(y)=0$ and $\indeg(y)=1$, consider $G-y-xy$. If $G-y-xy$ is neither a bush nor a double bush, then by inductive hypothesis for $k-1$ there is a successful $\LNS$-sequence $\sigma$ on $G-y-xy$. Let now $\tau$ be a sequence where all agents in $G$ except $y$ call $y$. We then have that and $\sigma;\tau$ is a successful $\LNS$ sequence on $G$. If however $G-y-xy$ is a bush, then by Lemma \ref{lemma-reverse-tree-1edge} we obtain a successful $\LNS$ sequence on $G$; and if $G-x-xy$ is a double bush, we obtain this by Lemma \ref{Lem:2JanPlusNode}. 

\item Let there be an edge $yx$ with $\indeg(y)=0$ and $\outdeg(y) = 1$ and $\outdeg(x)\geq 1$. If $G-y-yx$ is neither a bush nor a double bush, then by inductive hypothesis for $k-1$ there is a successful $\LNS$-sequence $\sigma$ on $G-y-xy$, which can be extended into a sequence $yx;\sigma;yt$ that is $\LNS$-successful on $G$, where $t$ is the successor of $x$ (see the figure above). Otherwise we proceed as before by applying Lemma \ref{lemma-reverse-tree-1edge} or Lemma \ref{Lem:2JanPlusNode}.

\item The remaining and more complex case is where, for any node $y$, if $\iodegr(y) = 1$, then: $\outdeg(y)=1$ and the successor $x$ of $y$ is a terminal; as above on the right. Provided that $x$ is terminal, since $G$ is not a bush, neither will be $G-y-yx$; similarly we conclude that $G-y-yx$ is not a double bush. So we can apply the inductive hypothesis to obtain some sequence $\sigma$ successful on $G-y-yx$. But the problem is to modify $\sigma$ into a solution for $G$. This case is proved in Lemma~\ref{lem:Case-2.2.3.2} which is found in the Appendix.
\end{itemize} \vspace{-.7cm}
\end{proof}

This finally brings us to our main result:
\begin{theorem}\label{conj-LNS}
$\LNS$ is weakly successful on a weakly connected gossip graph $G$ if{f} $G$ is neither a bush nor a double bush.
\end{theorem}
\begin{proof}
Directly, from the Propositions \ref{Thm:reversebush} (page \pageref{Thm:reversebush}), \ref{Lem:2Jan-GeneralCase} (page \pageref{Lem:2Jan-GeneralCase}), and \ref{theo.takki} (page \pageref{theo.takki}).
\end{proof}

\section{Conclusions and further research} \label{section.other}

\paragraph*{Conclusions}

We investigated distributed dynamic gossip protocols, where not only secrets are exchanged but also telephone numbers, such that the exchange of numbers in a call may expand the gossip graph. We considered the protocols $\SFS$ (any possible call), $\TOK$ (if you just called you are not permitted to call), $\SPI$ (if you were called you are not permitted to call), $\CO$ (if we called each other we may not call each other again), $\wCO$ (if you called me you may not call me again), and $\LNS$ (you are only permitted to call me if you do not know my secret). We characterized these protocols in terms of the class of gossip graphs where they terminate successfully. There are three such termination conditions: strong success (all protocol executions terminate with all agents knowing all secrets), fair success (strong success for fairly scheduled sequences), and weak success (at least one execution terminates). The results are as follows.

\begin{center}
\begin{tabular}{|c|c|c|c|}
\hline
 & $\SFS, \: \TOK, \: \SPI$ & $\CO, \: \wCO$ & $\LNS$ \\
\hline 
{\em strong/fair} &  weakly connected & weakly connected  & sun  \\
{\em success} & {\em Theorem \ref{Thm-SFS}, p.~\pageref{Thm-SFS}} & {\em Theorem \ref{calloncetheorem}, p.\ \pageref{calloncetheorem}} & {\em Theorem \ref{theorem.lnsstrong}, p.~\pageref{theorem.lnsstrong}} \\
\hline
{\em weak success} &  weakly connected &  weakly connected & no (bush or double bush) \\
& {\em Corollary \ref{cor.errare}, p.~\pageref{cor.errare}} & {\em Corollary \ref{cor.errare}, p.~\pageref{cor.errare}} & {\em Theorem \ref{conj-LNS}, p.~\pageref{conj-LNS}} \\
\hline
\end{tabular}
\end{center}

\paragraph*{Relevance} Our results may be considered relevant for different reasons. When defining a type of algorithm, such as the gossip protocol, it is satisfactory for a computer scientist (if not obligatory) to know what its termination conditions are. The interplay between very simple protocol conditions and rather complex emerging patterns on graphs, by their arbitrary iteration, may evoke a certain beauty for a combinatorial mathematician. Our results are clearly very remote from applications, but we hope networks researchers used to statistical methods may see scaling opportunities for such fairly novel, informed, gossip protocols. From a perspective of social network analysis, the difference between strong success, weak success, and lack of success can be interpreted as the assurance that the information dissemination process that you are about to enter may work out well no matter what, or has at least some chance to succeed if you watch your step carefully, or will utterly fail. In the last case, why even bother to start? But with a fair chance, you might as well give it a try.

\paragraph*{Further research} Various other distributed dynamic gossip protocols that we investigated, and that use the same protocol condition language, were all characterized by weakly connected graphs. We would be greatly interested in protocols characterized by truly different gossip graph topologies. One direction for further research are gossip protocols with protocol conditions that involve higher-order knowledge, for example: I will only call you if I know that you can inform me with a new secret in that call. Such a follow-up investigation, not focussed on characterization but on (a)synchronicity (absence or presence of a global clock), has appeared in \cite{DitmarschEPRS17}. Gossip protocols with unidirectional information change (push or pull, instead of pushpull), seem promising to find novel characterization graph topologies, as well as parallel gossip protocols with rounds of calls. 

\bibliographystyle{abbrv}
\bibliography{biblio,biblio2018}

\section*{Appendix}

We introduce some more terminology, to be used in Lemma \ref{lem:Case-2.2.3.2}. A {\em source} is a node with $\outdeg(s) \geq 2$. An {\em initial source} is a source that is a minimal node. A \emph{$t$-ghost} (or \emph{ghost}) is a node $x$ with $\indeg(x) = 0$, $\outdeg(x) = 1$, and $Nxt$ for some terminal node $t$. We introduce some notation for a path $x_1,\dots,x_n$ as $[x_1,x_n]$; and where $x \in [x_1,x_n]$ means $x \in \{x_1,\dots,x_n\}$. If we want to exclude the first node in the path we write $(x_1,x_n]$, and similarly for $[x_1,x_n)$ and for $(x_1,x_n)$ (unless the latter causes ambiguity with a pair in a relation). A {\em lonely path} is a path $[x_1,x_n]$ such that for all $y \in (x_1,x_n)$, the in-degree and out-degree of $y$ is 1. In other words, in a lonely path there is no branching in or branching out, except maybe at the first or at the last node.

Let $G = (\agents,N,S)$ be a gossip graph, $B \subseteq \agents$ and $x \notin \agents$. The $B$-\emph{abstraction} of $G$ is the graph $G' = (\agents', N', S')$ where: $\agents' = (\agents \setminus B)+x$; $N'yz$ iff ($Nyz$, or $x=y$ and $\exists w \in B$ such that $Nwz$, or $x=z$ and $\exists w \in B$ such that $Nyw$, or $x = y = z$); and similarly for $S'$. Informally, the $B$-abstraction replaces all nodes in $B$ by the single node $x$, and any edge from a node in $B$ to a node in $A \setminus B$ by an edge from $x$ to that node in $A \setminus B$, and vice versa.

\begin{lemma}\label{lem:Case-2.2.3.2}
Let $G$ be a weakly connected gossip graph where all edges are bridges, with at least one source node, wherein any node with inout-degree 1 is a ghost, and that is not a double bush. Then $G$ has a successful $\LNS$ sequence.
\end{lemma}
\begin{proof}
The proof is by induction on the number of source nodes. There are two base cases: for 1 source node, and for 2 source nodes in a particular configuration. The inductive case applies to 2 source nodes not in that configuration and to 3 or more source nodes. 

{\bf One source node} \ \ The treatment for a single source with two successors is different from the treatment for a single source with more than two successors. Figure \ref{fig.onesource} illustrates the two cases. Terminals are named $t_{1}, t_{2}, \ldots, t_{n}$ and the source is named $s$. 

\begin{figure}
\begin{tabular}{cc}
\fbox{
\begin{tikzpicture}
\node (empty) at (-2,4.645) {\quad};
\draw [fill=lightgray] (-3.5,0) circle (.25cm);
\draw[-,thick,fill=white,white] (-4,-.15) -- (-4,.1) -- (-3,.1) -- (-3, -.15) -- (-4,-.15);
\draw [fill=lightgray] (-4.5,3) circle (.25cm);
\draw[-,thick,fill=white,white] (-5,2.85) -- (-5,3.1) -- (-4,3.1) -- (-4, 2.85) -- (-5,2.85);
\draw [fill=lightgray] (-2.75,3) circle (.25cm);
\draw[-,thick,fill=white,white] (-3.25,2.85) -- (-3.25,3.1) -- (-2.25,3.1) -- (-2.25, 2.85) -- (-3.25,2.85);
\draw[-,thick,fill=white,white] (-4.35,.85) -- (-4.35,1.1) -- (-3.35,1.1) -- (-3.35, .85) -- (-4.35,.85);
\node (s1) at (-3.5,0) {\footnotesize source $s$};
\node (terminal1) at (-4.5, 3) {\footnotesize $t_1$};
\node (terminal1b) at (-2.75, 3) {\footnotesize $t_2$};
\node (ghost1) at (-6,2) {\footnotesize $t_1\mbox{-ghost}$};
\node (ghost1b) at (-5.5,4) {\footnotesize $t_1\mbox{-ghost}$};
\node (ghost1c) at (-2,4) {\footnotesize $t_2\mbox{-ghost}$};
\node (ghost1d) at (-1.5,2) {\footnotesize $t_2\mbox{-ghost}$};
\node at (-4.25,1.7) {\footnotesize path $1$};
\node at (-2.75,1.7) {\footnotesize path $2$};
\path [->, dashed]
(s1) edge node [left] {} (terminal1) 
(s1) edge node [left] {} (terminal1b) 
(ghost1) edge [bend left] node [left] {} (terminal1) 
(ghost1b) edge [bend left] node [left] {} (terminal1)
(ghost1c) edge [bend right] node [left] {} (terminal1b)
(ghost1d) edge [bend right] node [left] {} (terminal1b)
;
\end{tikzpicture} 
}
&
\fbox{
\begin{tikzpicture}
\draw [fill=lightgray] (-1,0) circle (.25cm);
\draw[-,thick,fill=white,white] (-1.5,-.15) -- (-1.5,.1) -- (-.5,.1) -- (-.5, -.15) -- (-1.5,-.15);
\draw [fill=lightgray] (-4.5,3) circle (.25cm);
\draw[-,thick,fill=white,white] (-5,2.85) -- (-5,3.1) -- (-4,3.1) -- (-4, 2.85) -- (-5,2.85);
\draw [fill=lightgray] (-2.75,3) circle (.25cm);
\draw[-,thick,fill=white,white] (-3.25,2.85) -- (-3.25,3.1) -- (-2.25,3.1) -- (-2.25, 2.85) -- (-3.25,2.85);
\draw [fill=lightgray] (-1,3) circle (.25cm);
\draw[-,thick,fill=white,white] (-1.5,2.85) -- (-1.5,3.1) -- (-.5,3.1) -- (-.5, 2.85) -- (-1.5,2.85);
\draw [fill=lightgray] (1,3) circle (.25cm);
\draw[-,thick,fill=white,white] (1.5,2.85) -- (1.5,3.1) -- (.5,3.1) -- (.5, 2.85) -- (1.5,2.85);

\node (s1) at (-1,0) {\footnotesize source $s$};
\node (terminal1) at (-4.5, 3) {\footnotesize $t_1$};
\node (terminal1b) at (-2.75, 3) {\footnotesize $t_2$};
\node (sharedterminal) at (-1, 3) {\footnotesize $\cdots$}; 
\node (sharedterminalb) at (1, 3) {\footnotesize $t_n$}; 

\node (xx) at (-5.5,1.6) {\footnotesize $x$};
\node (ghost1) at (-5.5,2) {\footnotesize $t_{1}\mbox{-ghost}$};
\node (ghost1b) at (-5.5,4) {\footnotesize $t_{1}\mbox{-ghost}$};
\node (ghost1c) at (-4.55,4.5) {\footnotesize $t_{2}\mbox{-ghost}$};
\node (ghost1d) at (-3,4.5) {\footnotesize $t_{2}\mbox{-ghost}$};

\node (sharedghost1b) at (-1.5,4.5) {\footnotesize $t_{i}\mbox{-ghost}$};
\node (sharedghost1c) at (0,4.5) {\footnotesize $t_{i}\mbox{-ghost}$};
\node (sharedghost1d) at (1.5,4.5) {\footnotesize $t_{n}\mbox{-ghost}$};

\node at (-3.9,1.5) {\footnotesize path 1};
\node at (-2.5,1.5) {\footnotesize path 2};
\node at (.9,1.5) {\footnotesize path $n$};

\path [->, dashed]
(s1) edge [bend left] node [left] {} (terminal1) 
(s1) edge [bend left] node [left] {} (terminal1b) 
(s1) edge [bend right] node [left] {} (sharedterminal)
(s1) edge [bend right] node [left] {} (sharedterminalb) 

(ghost1) edge [bend left] node [left] {} (terminal1) 
(ghost1b) edge [bend left] node [left] {} (terminal1)
(ghost1c) edge [bend left] node [left] {} (terminal1b)
(ghost1d) edge [bend left] node [left] {} (terminal1b)

(sharedghost1b) edge [bend left] node [left] {} (sharedterminal)
(sharedghost1c) edge [bend left] node [left] {} (sharedterminal)
(sharedghost1d) edge [bend left] node [left] {} (sharedterminalb);
\end{tikzpicture} 
} 
\end{tabular}
\caption{Gossip graphs with 1 source node and with 2 resp.\ more than 2 successors}
\label{fig.onesource}
\end{figure}
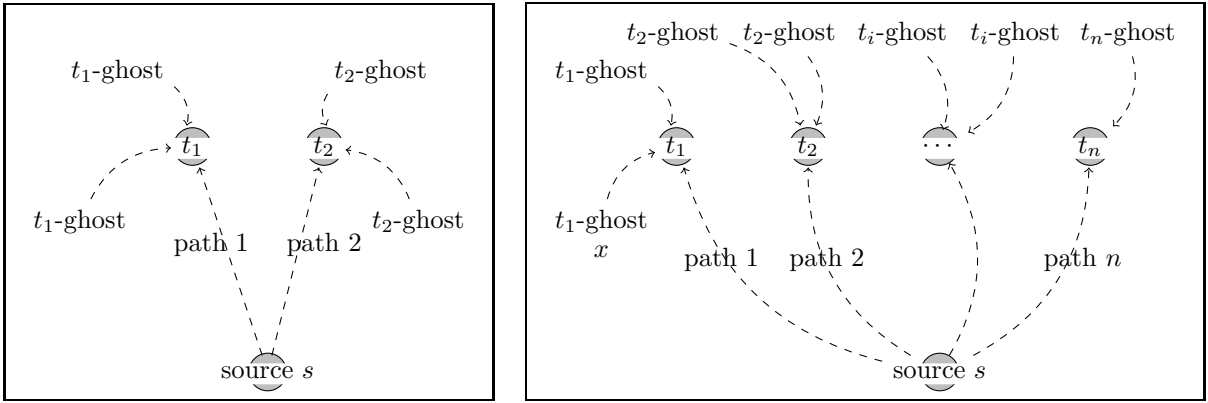

If $\outdeg(s)=2$ at least one path must have length 2, because otherwise the graph is a double bush; say this is on the path $[s,t_1]$. Note that $(s,t_1)$ is then non-empty; this plays a role below. The following call sequence is $\LNS$-permitted and successful on $G$:
\begin{itemize}
\item a maximal bottom-up call sequence for $[s, t_{2})$;
\item a maximal bottom-up call sequence for $[s, t_{1})$;
\item all agents in path $(s,t_{1})$ call $t_{1}$;
\item each $t_{1}$-ghost calls $t_{1}$;
\item call $st_{2}$; each $t_{2}$-ghost calls $t_{2}$;
\item call $t_{1} t_{2}$; all agents in path $(s, t_{1})$ and all $t_{1}$-ghosts call $t_{2}$;
\item all agents in path $(s,t_{2})$ call $t_{2}$;
\item all $t_2$-ghosts call $t_1$; call $s t_{1}$.
\end{itemize}
If $\outdeg(s)\geq 3$, fix $x$ to be some $t_1$-ghost. The following is $\LNS$-successful on $G$: 
\begin{itemize}
\item a maximal bottom-up call sequence for path $[s,t_2)$;
\item a maximal bottom-up call sequence for path $[s,t_1]$;
\item each $t_1$-ghost calls $t_1$;
\item call $st_{2}$; each $t_2$-ghost calls $t_2$;
\item for $3 \leq i \leq n$: $t_1$ calls any agent along the path $(s,t_{i}]$; each $t_i$-ghost calls $t_i$; $x$ calls any agent along the path $(s,t_{i}]$;
\item $t_n$ calls $t_2$; all agents except $s$, $t_n$, $t_2$, and $t_2$-ghosts call $t_2$;
\item the source $s$ and each $t_2$-ghost call the successor of $s$ in path $n$.
\end{itemize}
{\bf Two source nodes with special condition} \ \  Given tree $G_1$ with root $t_1$ and tree $G_2$ and root $t_2$, consider gossip graph $G$ consisting of $G_1$ and $G_2$ plus a source $s_1$ that is connected by a link to $t_1$ and by a path to a terminal $t$, and a source $s_2$ that is connected by a link to $t_2$ and also by a path to $t$.

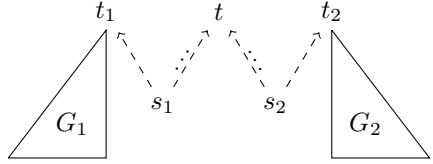
\begin{center}
\begin{tikzpicture}
\node (t1) at (-.75,.75) {\footnotesize $t_1$};
\node (s1) at (0,-.5) {\footnotesize $s_1$};
\node (t) at (.75,.75) {\footnotesize $t$};
\node (s2) at (1.5,-.5) {\footnotesize $s_2$};
\node (t2) at (2.25,.75) {\footnotesize $t_2$};
\node (G1) at (-1.2,-.75) {\footnotesize $G_1$};
\node (G2) at (2.7,-.75) {\footnotesize $G_2$};

\node (cdots1) at (0.23,-0.08) {\scriptsize \begin{rotate}{58}$\cdots$\end{rotate}}; 
\node (cdots2) at (1.38,-0.02) {\scriptsize \begin{rotate}{120}$\cdots$\end{rotate}};

\draw[-] (-2.05,-1.2) -- (-.75,-1.2) -- (-.75,.5) -- (-2.05,-1.2);
\draw[-] (2.25,.5) -- (2.25,-1.2) -- (3.55,-1.2) -- (2.25,.5);
\draw[dashed, ->] (s1) -- (t1); 
\draw[dashed, ->] (s1) -- (t); 
\draw[dashed, ->] (s2) -- (t); 
\draw[dashed, ->] (s2) -- (t2); 
\end{tikzpicture}
\end{center}
Let $\sigma^1$ and $\sigma^2$ be maximal bottom-up call sequences in,
respectively, $G_1$ and $G_2$. The sequence
\[ [s_1t]; [s_2t]; tt_1; s_{2}t_{2}; \sigma^1; \sigma^2; t_2t_1; [s_{1},t) t_{1}; [s_{2}, t) t_{1}; tt_{2};\tau^1; \tau^2 \]
is $\LNS$-successful on $G$, where $[s_1t]$ is a maximal bottom-up sequence for path $[s_1,t]$, and similarly for $[s_2t]$, where $[s_1,t)t_{1}$ is a sequence wherein each agent in $[s_1,t)$ calls $t_1$ (and $[s_2,t)t_{1}$ is defined similarly); finally, $\tau^1$ is the sequence such that all the agents in $G_1$ except $t_1$ call $t_2$; $\tau^2$ is the sequence that all the agents in $G_2$ except $t_2$ call $t_1$.

{\bf (At least two source nodes)} 

We first assume that there are two initial source nodes $s_1$ and $s_2$. Let $u$ be the node where the paths from $s_1$ and $s_2$ meet (so $[s_1,u]$ and $[s_2,u]$ are lonely paths---no branching in, no branching out). See also Figure \ref{alweereenfigure}.

\begin{figure}
\begin{center}
\begin{tabular}{@{\hspace{4mm}}c@{\hspace{20mm}}c@{\hspace{7mm}}}
\scalebox{.8}{
\begin{tikzpicture}[scale=.8]				%
\node (g-name) at (0, 3.36) {$G$};
\node (a) at (3,1.15) {\textbullet};
\node (a2) at (3, 0.3) {\textbullet};
\node (b) at (1,1) {$v$};
\node (c) at (1.7,0) {$s_1$};
\node (u) at (3,2) {$u$};
\node (e) at (4.3,0) {$s_2$};
\node (f) at (5,1) {\textbullet};
\node (g) at (2.2,2.5) {\textbullet};
\node (zz1) at (.75, 1.3) {\begin{rotate}{30} $\vdots$ \end{rotate}};
\node (zz2) at (3, 3.75) {$\vdots$};
\node (zz3) at (5.1, 1.3) {\begin{rotate}{-30} $\vdots$ \end{rotate}};
\node (zz4) at (3, 0) {$\vdots$};
\node (g2) at (3, 3.2) {\textbullet};
\draw[dashed, ->] (a) -- (u); 
\draw[dashed, ->] (a2) -- (a); 
\draw[dashed, ->] (c) -- (b); 
\draw[thick, dashed, ->] (c) to [bend right] (2.75,1.7); 
\draw[thick, dashed, ->] (e) to [bend left] (3.25,1.7); 
\draw[dashed, ->] (e) -- (f); 
\draw[dashed, ->] (u) -- (g2); 
\draw[dashed, ->] (g) to [bend left] (2.8, 2.3); 
\end{tikzpicture} 
}
&
\scalebox{.8}{
\begin{tikzpicture}[scale=.8]					%
\node (g-name) at (1, 3) {$G$};
\node (v) at (1,.5) {$v$};
\node (c) at (2,-0.5) {$s_1$};
\node (d) at (5,1) {$s_2$};
\node (e) at (4,2) {\textbullet};
\node (f) at (6,2) {\textbullet};
\node (h) at (6,0) {\textbullet};
\node (zz1) at (.75, .7) {\begin{rotate}{30} $\vdots$ \end{rotate}};
\node (zz2) at (3.8, 2.2) {\begin{rotate}{30} $\vdots$ \end{rotate}};
\node (zz3) at (6.1, 2.25) {\begin{rotate}{-30} $\vdots$ \end{rotate}};
\node (zz4) at (6.2, -0.5) {\begin{rotate}{30} $\vdots$ \end{rotate}};
\draw[dashed, ->] (c) -- (v); 
\draw[thick, dashed, ->] (c) to [bend right] (d); 
\draw[dashed, ->] (d) -- (e); 
\draw[dashed, ->] (d) -- (f); 
\draw[dashed, ->] (h) -- (d); 
\end{tikzpicture}
} \\
%
\hspace{.7cm} $\Downarrow$ &\hspace{1.6cm} $\Downarrow$ \\
\scalebox{.8}{
\begin{tikzpicture}[scale=.8]					%
\node (g-name) at (0, 3.56) {$H$};
\draw [fill=white] (3,2) circle (.28cm);
\node (a) at (3,1.15) {\textbullet};
\node (a2) at (3, 0.3) {\textbullet};
\node (b) at (1,2.5) {$v$};
\node (u) at (3,2) {$x$};
\node (f) at (5,2.5) {\textbullet};
\node (zz1) at (.75, 2.8) {\begin{rotate}{30} $\vdots$ \end{rotate}};
\node (zz2) at (3, 3.75) {$\vdots$};
\node (zz3) at (5.1, 2.8) {\begin{rotate}{-30} $\vdots$ \end{rotate}};
\node (zz4) at (3, 0) {$\vdots$};
\node (g2) at (3, 3.2) {\textbullet};
\draw[dashed, ->] (a) -- (u); 
\draw[dashed, ->] (a2) -- (a); 
\draw[dashed, ->] (u) -- (b); 
\draw[dashed, ->] (u) -- (f); 
\draw[dashed, ->] (u) -- (g2); 
\end{tikzpicture} 
}
&
\scalebox{.8}{
\begin{tikzpicture}[scale=.8]					%
\node (g-name) at (1, 3.2) {$H$};
\draw [fill=white] (5,1) circle (.28cm);
\node (v) at (3.5,1.25) {$v$};
\node (d) at (5,1) {$x$};
\node (e) at (4,2) {\textbullet};
\node (f) at (6,2) {\textbullet};
\node (h) at (6,0) {\textbullet};
\node (zz1) at (3.3, 1.45) {\begin{rotate}{30} $\vdots$ \end{rotate}};
\node (zz2) at (3.8, 2.2) {\begin{rotate}{30} $\vdots$ \end{rotate}};
\node (zz3) at (6.1, 2.25) {\begin{rotate}{-30} $\vdots$ \end{rotate}};
\node (zz4) at (6.2, -0.5) {\begin{rotate}{30} $\vdots$ \end{rotate}};
\draw[dashed, ->] (d) -- (v); 
\draw[dashed, ->] (d) -- (e); 
\draw[dashed, ->] (d) -- (f); 
\draw[dashed, ->] (h) -- (d); 
\end{tikzpicture}
}
\end{tabular}
\end{center}
\caption{On the left: two initial sources. On the right: an initial and a non-initial source}
\label{alweereenfigure}
\end{figure}
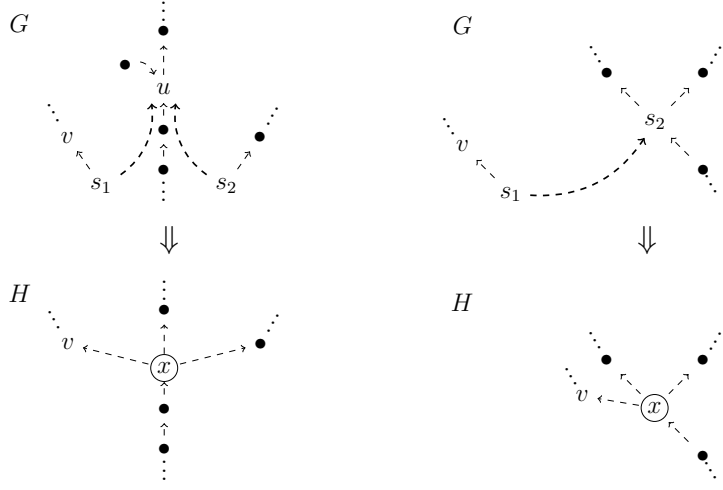

Let $\sigma^{1}$ and $\sigma^{2}$ be maximal bottom-up call sequences in, respectively, $[s_1, u]$ and $[s_2, u]$. If $u$ is a terminal, then let $\sigma^{3}$ be a call sequence from each $u$-ghost to $u$; otherwise $\sigma^{3} = \epsilon$. Consider the $([s_1,u]\union[s_2,u]\union\{ z \mid z \mbox{ is a $u$-ghost}\})$-abstraction of $G$ and let $x$ be the new node generated by this abstraction. Call this graph $G'$. Graph $G'$ has one less source. If $G'$ is a double bush, the base case {\bf (Two source nodes with special condition)} applies and there is a successful $\LNS$-sequence $\sigma$ on $G'$. Otherwise, there is such a $\sigma$ on $G'$ by inductive hypothesis.

Let now $v$ be a successor of $s_1$ that is not in the path $[s_1, u]$. Let $\tau$ be obtained from $\sigma$ by replacing each occurrence of $x$ with $u$, and let $\tau'$ be a call sequence where all agents in $[s_1,u) \union [s_2,u)$, and also all $u$-ghosts, call $v$ (that is, all the agents different from $u$ removed in the abstraction call $v$.)

Then a successful $\LNS$-sequence on $G$ is $\sigma^{1}; \sigma^{2}; \sigma^{3}; \tau; \tau'$.

We can further justify why it is successful. After $\sigma^{1}; \sigma^{2}; \sigma^{3}$, agent $u$ knows the secret of all agents in $[s_1,u] \union [s_2,u] \union \{z \mid z \mbox{ is a $u$-ghost}\}$. After $\tau$, all agents in $G$ are experts except those in the abstracted part that are not $u$. After $\tau'$, those other agents also become experts.

If there are no two initial sources, then, given that there are at least two sources, there must be two sources $s_1, s_2$ such that $s_1$ is an initial source and $s_1$ is connected to $s_2$ by a lonely path (see Figure \ref{alweereenfigure}). Let $\sigma^{1}$ be a maximal bottom-up call sequence in $[s_1, s_2]$. Consider the $[s_1,s_2]$-abstraction of $G$ and let $x$ be the new node. Let this graph be $H$. Since $H$ has one less source node (and since it is not a double bush), there is an $\LNS$-sequence $\sigma$ successful on $H$. Let $\tau$ be the call sequence obtained from $\sigma$ where we replace each occurrence of $x$ with $s_2$, and let $\tau'$ be the call sequence where every agent in $[s_1, s_2)$ calls a successor $v$ of $s_1$ not in $[s_1, s_2]$. Similarly to the previous case, it can be easily shown that $\sigma^{1}; \tau; \tau'$ is an $\LNS$-sequence successful on $G$. 
\end{proof}

\end{document}